\documentclass[runningheads]{llncs-thmctr}

\usepackage{graphicx}
\usepackage{amsmath}
\usepackage{amssymb}
\usepackage{latexsym}
\usepackage{xspace}
\usepackage{stmaryrd}
\usepackage{cite}
\usepackage{wrapfig}
\usepackage{color,colortbl}
\definecolor{LightCyan}{rgb}{0.88,1,1}
\usepackage[inline]{enumitem}

\usepackage{mathtools}
\usepackage[breaklinks]{hyperref}
\hypersetup{
colorlinks = true, 
linkcolor=blue, 
citecolor=blue, 
urlcolor=blue, 
bookmarks, bookmarksdepth=2 
}
\usepackage[capitalize,nameinlink]{cleveref}
\usepackage{algorithm}
\usepackage{algpseudocode}
\usepackage[disable]{todonotes}

\usepackage{multirow}
\usepackage{array, booktabs}
\usepackage{threeparttable}
\usepackage{makecell}

\usepackage{tikz}
\usetikzlibrary{arrows.meta,arrows,calc,matrix,shapes,decorations,automata,backgrounds,positioning}
\usepackage{pgfplots}
\usepackage{fontawesome5}

\usepackage{orcidlink} 

\newcommand{\Signal}{\mathbf{Signal}}
\newcommand{\Sub}{\mathrm{Sub}}
\newcommand{\AP}{\mathrm{AP}}
\newcommand{\Robust}[2]{\ensuremath\sem{#1, #2}\xspace}
\newcommand{\BoolSat}[2]{\ensuremath #1\models #2\xspace}
\newcommand{\R}{\mathbb{R}}
\newcommand{\dimNum}{\ensuremath{V}\xspace}
\newcommand{\UntilOp}[1]{\mathbin{\mathcal{U}_{#1}}}
\newcommand{\ReleaseOp}[1]{\mathbin{\mathcal{R}_{#1}}}
\newcommand{\Defeq}{\coloneqq}
\newcommand{\signalSymbol}{\ensuremath{\sigma}\xspace}
\newcommand{\sem}[1]{\llbracket{#1}\rrbracket}
\newcommand{\pow}{\wp}
\newcommand{\Defiff}{\Longleftrightarrow}
\newcommand{\Real}{\mathbb{R}}
\newcommand{\EReal}{\overline{\mathbb{R}}}

\newcommand{\front}{\mathrm{f}}
\newcommand{\rear}{\mathrm{r}}
\newcommand{\Carf}{\mathrm{Car}_{\front}}
\newcommand{\Carr}{\mathrm{Car}_{\rear}}
\newcommand{\MRNC}{\mathcal{M}_{\mathrm{RNC}}}
\newcommand{\MNAV}{\mathcal{M}_{\mathrm{NAV}}}
\newcommand{\hl}{\cellcolor{green!20}}

\newcommand{\myparagraph}[1]{\medskip\noindent{\bf #1.}}
\allowdisplaybreaks

\spnewtheorem{theorem}{Theorem}[section]{\bfseries}{\rmfamily}

\spnewtheorem{definition}[theorem]{Definition}{\bfseries}{\rmfamily}

\spnewtheorem{corollary}[theorem]{Corollary}{\bfseries}{\rmfamily}

\spnewtheorem{example}[theorem]{Example}{\bfseries}{\rmfamily}

\spnewtheorem{lemma}[theorem]{Lemma}{\bfseries}{\rmfamily}

\spnewtheorem{remark}[theorem]{Remark}{\bfseries}{\rmfamily}

\spnewtheorem{problem}[theorem]{Problem}{\bfseries}{\rmfamily}

\spnewtheorem{proposition}[theorem]{Proposition}{\bfseries}{\rmfamily}
\spnewtheorem{assumption}[theorem]{Assumption}{\bfseries}{\rmfamily}
\spnewtheorem{notation}[theorem]{Notation}{\bfseries}{\rmfamily}

\creflabelformat{enumi}{#2(#1)#3}
\crefname{theorem}{Thm.}{Thms}
\crefname{definition}{Def.}{Defs}
\crefname{proposition}{Prop.}{Props}
\crefname{remark}{Rem.}{Remarks}
\crefname{lemma}{Lem.}{Lemmas}
\crefname{example}{Ex.}{Exs}
\crefname{problem}{Prob.}{Probs}
\crefname{proof}{Proof.}{Proofs}
\crefname{appendix}{Appendix}{Appendixes}
\crefname{algorithm}{Alg.}{Algs}
\crefformat{section}{{\S}#2#1#3}
\crefname{figure}{Fig.}{Figs}
\Crefname{equation}{}{}

\usepackage{comment}
\specialcomment{auxproof}
{\mbox{}\newline\textbf{BEGIN: AUX-PROOF}\dotfill\newline}
{\mbox{}\newline\textbf{END: AUX-PROOF}\dotfill\newline}
\excludecomment{auxproof}

\begin{document}
\title{
    \texorpdfstring{Optimization-Based Model Checking and\\ Trace Synthesis for Complex STL Specifications\\ (Extended Version)}{Optimization-Based Model Checking and Trace Synthesis for Complex STL Specifications (Extended Version)}\thanks{The authors are supported by ERATO HASUO Metamathematics for Systems
Design Project (No. JPMJER1603), the  START Grant No.\ JPMJST2213, the ASPIRE grant No.\ JPMJAP2301, JST. S.S.\ is supported by KAKENHI No.\ 23KJ1011, JSPS. Z.Z.\ is supported by JSPS KAKENHI Grant No.\ JP23K16865 and No.\ JP23H03372.
}
}
\titlerunning{Optimization-Based Trace Synthesis for Complex STL Specifications}
%
\author{
    Sota Sato \inst{1,3}\,\orcidlink{0000-0001-7147-3989} \and
    Jie An \inst{1,4}\,\orcidlink{0000-0001-9260-9697} \and
    Zhenya Zhang \inst{2,1}\,\orcidlink{0000-0002-3854-9846} \and
    Ichiro Hasuo \inst{1,3}\,\orcidlink{0000-0002-8300-4650}
}
\institute{
        National Institute of Informatics, Tokyo, Japan \\
        \email{\{sotasato,jiean,hasuo\}@nii.ac.jp}
    \and
        Kyushu University, Fukuoka, Japan \\
        \email{zhang@ait.kyushu-u.ac.jp}
    \and SOKENDAI (The Graduate University for Advanced Studies), Tokyo, Japan
    \and
        Institute of Software, Chinese Academy of Sciences, Beijing, China \\
}
\authorrunning{
    S. Sato et al.
}
%

%
\maketitle              

\begin{abstract}
Techniques of light-weight formal methods, such as monitoring and falsification,  are attracting attention for quality assurance of cyber-physical systems. The techniques  require formal specs, however, and  writing  right specs
is still  a practical challenge. Commonly one relies on \emph{trace synthesis}---i.e.\ automatic generation of a signal that satisfies a given spec---to examine the meaning of a spec.
In this work, motivated by 1) complex STL specs from an automotive safety standard and 2) the struggle of existing tools in their trace synthesis, we introduce a novel trace synthesis algorithm for STL specs. It combines the use of MILP (inspired by works on controller synthesis) and a \emph{variable-interval encoding} of STL semantics (previously studied for SMT-based STL model checking). The algorithm solves model checking, too, as the dual of trace synthesis. Our experiments show that only ours has realistic performance needed for the interactive examination of STL specs by trace synthesis.

\end{abstract}


%
%
\section{Introduction}\label{sec:intro}
Safety and quality assurance of \emph{cyber-physical systems (CPSs)} is an important and multifaceted problem. The pervasiveness and safety-critical nature of CPSs makes the problem imminent and pressing; at the same time, the problem comes with very different flavors in different application domains, calling for  different solutions. For example, in the aerospace domain, full  formal verification all the way up from the codebase seems  feasible~\cite{SouyrisWDD09}. Such is a luxury that the automotive domain may not afford, however, because of short  product cycles, dependence on third-party (thus black-box) components, heterogeneous environmental uncertainties, and fierce competition (thus tight budget).

The above limitations in the automotive domain point, in the formal methods terms, to the \emph{absence of white-box system models}. This has led to the flourish of \emph{light-weight formal methods}, such as monitoring~\cite{BartocciDDFMNS18}, runtime verification, and hybrid system falsification~\cite{ErnstABCDFFG0KM21}. These are logic-based methods that operate on \emph{formal specifications}, often given in \emph{signal temporal logic (STL)}~\cite{MalerN04}.
These methods give up comprehensive guarantee due to the absence of white-box system models; yet their values in practical usage scenarios are widely acknowledged.

\myparagraph{Trace Synthesis and Model Checking}
In this paper, we are motivated by some automotive instances of  the \emph{trace synthesis} problem: it asks to synthesize an execution trace $\sigma$ of a system $\mathcal{M}$ that satisfies a given STL specification $\varphi$.
There are two major approaches to trace synthesis for CPSs.

One common approach is via \emph{hybrid system falsification}~\cite{ErnstABCDFFG0KM21}: here, we try  many input signals $\tau$  for $\mathcal{M}$, iteratively modifying them in the direction of satisfying $\varphi$; the quantitative \emph{robust semantics} of STL~\cite{fainekosRobustnessTemporalLogic2009} serves as an objective function that allows hill-climbing optimization. It is notable that the system model $\mathcal{M}$ can be \emph{black-box}: we do not need to know its internal working; it is enough to compute the execution trace $\mathcal{M}(\tau)$ under given input $\tau$. Falsification has  attracted a lot of interest especially in the automotive domain; see e.g.~\cite{ErnstABCDFFG0KM21}.

We take the other approach to trace synthesis, namely as the \emph{dual of the model checking problem}. Here model checking decides if, under \emph{any} input $\tau$, the execution trace   $\mathcal{M}(\tau)$ satisfies $\varphi $. Our choice of this approach may be puzzling---it requires a white-box model $\mathcal{M}$, but  it is rare in the automotive domain.


\myparagraph{Analyzing Specifications (Rather Than Models)}
Our choice of the model checking approach to trace synthesis comes from the following basic scope of the paper: \emph{we use trace synthesis to analyze the quality of  specifications (specs)}.

This is in stark contrast with  many falsification tools whose scope is analyzing \emph{models}. There, a model $\mathcal{M}$ is extensive and complex (typically a Simulink model of an actual product), and  counterexample traces are  used for ``debugging'' $\mathcal{M}$.

In this paper, instead, a model $\mathcal{M}$ is simple and white-box (it can even be the trivial model, where the input and output are the same), but a spec $\varphi$ tends to be complex. One typical usage scenario for our framework is when $\varphi$ is a \emph{normative rule}---such as a law, a traffic rule, or a property required in an international standard---in which case $\varphi$ is imposed on many different systems (e.g.\ different vehicle models). Then $\mathcal{M}$ should be a simple overapproximation of a variety of systems, rather than a detailed system model.

Another typical usage scenario of our framework  is an early ``requirement development'' phase of the \emph{V-model} of the automotive system design.
Here, engineers fix specs  that pin down the  later development efforts, in which those specs get refined and realized. They want to confirm that the specs are sensible (e.g.\ there is no mutual conflict) and faithful to their intentions. Since a system is yet to be developed, a system model $\mathcal{M}$ cannot be detailed.

\myparagraph{Motivating Example}
More specifically, the current work is motivated by the work~\cite{ReimannMHBCSWAHUY24_toAppear} on formalizing disturbance scenarios in the ISO 34502 standard for automated driving vehicles. There, a vehicle dynamics model is simple (the scenarios should apply to different vehicle models---see above), but STL formulas are complex. It is observed that existing algorithms have a hard time handling the complexity of specs (see \cref{sec:expr} for experiments). This motivated our current technical development, namely a trace synthesis algorithm that exploits \emph{white-box models} and \emph{MILP optimization} for efficiency.

The following example
illustrates
 the challenge encountered in~\cite{ReimannMHBCSWAHUY24_toAppear}.

\begin{wrapfigure}[5]{r}[0pt]{4.5cm}
    \centering

    \vspace{-2em}
    \includegraphics[width=3.5cm]{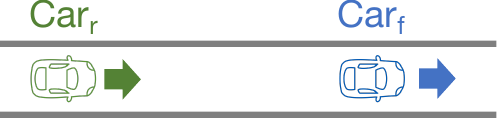}
    \caption{Rear-end near collision}
    \label{fig:oneway}
\end{wrapfigure}

\begin{example}[rear-end near collision]\label{ex:introLead}
    We would like to express, in STL, a \emph{rear-end near collision} scenario for two cars. It refers to those driving situations where a rear car $\Carr$ comes too close to a front car $\Carf$. We assume a single-lane setting (\cref{fig:oneway}), so we can ignore  lateral dynamics.

    Consider the following STL formulas. Here, $x_{\front}, v_{\front}, a_{\front}$ are the variables for the position, velocity, and acceleration of $\Carf$; the other variables are for $\Carr$.
    \begin{align}\label{eq:RNC1}
        \begin{array}{rl}
            \mathtt{danger}     & \quad:\equiv\quad x_{\front} - x_{\rear} \leq 10                                                                                                \\
            \mathtt{dyn\_{inv}} & \quad:\equiv\quad x_{\front} - x_{\rear} \ge 0\,\land\, 2 \leq v_{\front} \leq 27 \,\land\, 2 \leq v_{\rear} \leq 27
            \\
            \mathtt{trimming}   & \quad:\equiv\quad (\Diamond \mathtt{danger}) \Rightarrow \bigl((\Box_{[0, 0.2]} a_{\rear} \geq 0.5) \mathbin{\mathcal{U}} \mathtt{danger}\bigr)
            \\
            \mathtt{RNC1}       & \quad:\equiv\quad \Box(\mathtt{dyn\_inv}\land \mathtt{trimming}) \land \Diamond_{[0, 9]}\Box_{[0,1]}\mathtt{danger}
        \end{array}
    \end{align}
    The last formula $\mathtt{RNC1}$ formalizes rear-end near collision; in particular,  its subformula $\Diamond_{[0, 9]}\Box_{[0,1]}\mathtt{danger}$ requires that $\mathtt{danger}$ occurs within 9 seconds and persists for at least one second.

    The  formula $\mathtt{RNC1}$ comes with two auxiliary conditions: $\mathtt{dyn\_inv}$ and $\mathtt{trimming}$. We shall now exhibit their content and why they are needed. In fact, these conditions arose naturally in the course of
    \emph{trace synthesis}, the problem of our focus.

    Specifically, in~\cite{ReimannMHBCSWAHUY24_toAppear}, we conducted  trace synthesis repeatedly in order  to 1) \emph{illustrate} the meaning of STL specifications and 2) \emph{confirm} that they reflect informal intentions. The generated traces were  animated for graphical illustration. This workflow is much like in the tool STLInspector~\cite{RoehmHM17}.

    The formula $\mathtt{dyn\_inv}$ imposes basic constraints on the dynamics of the cars.
    In the trace synthesis in~\cite{ReimannMHBCSWAHUY24_toAppear},  without this basic  constraint, we obtained a number of nonsensical example traces in which a car warps and instantly passes the other, drives much faster than the legal maximum, and so on.

    The formula $\mathtt{trimming}$
    requires  $\Carr$ to accelerate until  $\mathtt{danger}$ occurs. It was added to limit a generated trace to an interesting part. For example, a trace can have $\mathtt{danger}$ only after a $8$-second pacific journey; animating this whole trace can easily bore users. The condition
    trims such a trace to the part where $\Carr$ is accelerating towards $\mathtt{danger}$.







    The dynamics model  used
    in~\cite{ReimannMHBCSWAHUY24_toAppear} is the following simple one:
    \begin{equation}\label{eq:introLeadModel}
        \dot{x_{\front}}=v_{\front}, \;
        \dot{v_{\front}}=a_{\front}; \qquad
        \dot{x_{\rear}}=v_{\rear}, \;
        \dot{v_{\rear}}=a_{\rear}. \;
    \end{equation}
    This  relates $x,v$ and $a$ in the spec~\cref{eq:RNC1}. The double integrator model is certainly simplistic, but it suffices the purpose in~\cite{ReimannMHBCSWAHUY24_toAppear} of illustrating and confirming specs.
\end{example}

\begin{remark}
    In~\cite{ReimannMHBCSWAHUY24_toAppear}, after illustrating and confirming STL specs through trace synthesis, the final goal was to use them for monitoring actual driving data. Neither the dynamics model \cref{eq:introLeadModel} nor the condition $\mathtt{dyn\_inv}$ is really relevant to monitoring---actual driving data should comply with them anyway. In contrast, $\mathtt{trimming}$ is important, in order to extract only relevant parts of the data.
\end{remark}

\myparagraph{Technical Solution: MILP-Based Trace Synthesis}
We present a novel  trace synthesis algorithm.  Note that it  also solves the dual problem, namely STL model checking. It originates from two recent lines of work:  MILP-based optimal control~\cite{ramanReactiveSynthesisSignal2015,RamanDMMSS14,donzeBluSTLControllerSynthesis} and SMT-based STL model checking~\cite{baeBoundedModelChecking2019,leeEfficientSMTBasedModel2021,YuLB22}.

The controller synthesis techniques in~\cite{ramanReactiveSynthesisSignal2015,RamanDMMSS14,donzeBluSTLControllerSynthesis} exploit \emph{mixed-integer linear programming (MILP)} for  efficiency. The optimal control problem that they solve can be specialized to our trace synthesis problem (detailed discussions come later). But we found  their capability of handling complex specs (as in \cref{ex:introLead})  limited, largely because of their \emph{constant-interval encoding} to MILP.

We solve this challenge by our novel \emph{variable-interval encoding} of the STL semantics to MILP. It is inspired by the \emph{stable partitioning} technique introduced in~\cite{baeBoundedModelChecking2019}: the technique is used in~\cite{baeBoundedModelChecking2019,leeEfficientSMTBasedModel2021,YuLB22}  for \emph{logical} encoding  towards SMT-based model checking; we use it for \emph{numerical} encoding to MILP.
This way we will solve the \emph{bounded} trace synthesis problem---in the sense that variability of the truth values of the relevant formulas is bounded---much like in~\cite{baeBoundedModelChecking2019,leeEfficientSMTBasedModel2021,YuLB22}.
For our MILP encoding, however, we need special care since MILP does not accommodate strict inequalities (partitions  such as
$\dotsc, (\gamma_{i-1},\gamma_{i}),\{\gamma_{i}\}, (\gamma_{i},\gamma_{i+1}), \dotsc$
in~\cite{baeBoundedModelChecking2019} cannot be expressed). We therefore use a novel technique called \emph{$\delta$-stable partitioning}.

Overall, our algorithm works as follows. We assume that a system model $\mathcal{M}$ can be MILP-encoded, either exactly or approximately. Some model families are discussed in \cref{sec:systemModels}. This assumption, combined with our key technique of \emph{variable-interval MILP encoding} of STL, reduces trace synthesis to an MILP problem, which we solve by Gurobi Optimizer~\cite{gurobi}. We conduct experimental evaluation to confirm the scalability of our algorithm, especially for complex specs (\cref{sec:expr}).

Our algorithm is \emph{anytime} (i.e.\ \emph{interruptible}): even if the budget runs out in the course of optimization, a best-effort result (the trace that is the closest to a solution so far) is obtained. A similar benefit is there in case there is no execution trace $\sigma$ that satisfies the spec $\varphi$: we obtain a trace $\sigma'$ that is the closest to satisfy $\varphi$. Accommodation of \emph{parameters} is another advantage thanks to our use of MILP; we exploit it for \emph{parameter mining} for PSTL formulas. See~\cref{sec:problemSetting}.

Both  controller synthesis techniques~\cite{ramanReactiveSynthesisSignal2015,RamanDMMSS14,donzeBluSTLControllerSynthesis} and SMT-based model checking techniques~\cite{baeBoundedModelChecking2019,leeEfficientSMTBasedModel2021,YuLB22} can be used for trace synthesis. The methodological differences are discussed later in \cref{sec:intro}; experimental comparison is made in \cref{sec:expr}.

\myparagraph{Contributions and Organization} We summarize our contributions.
\begin{itemize}
    \item We introduce an optimization-based  algorithm for bounded trace synthesis for STL specs. It assumes that a system model is white-box and  MILP-encodable; it also solves the dual problem (namely bounded model checking).
    \item As a key element, we introduce  a \emph{variable-interval} encoding of STL to MILP.
    \item  MILP encodings of some system models, notably rectangular hybrid automata and double integrator dynamics (suited for the automotive domain).
    \item We experimentally confirm scalability of our algorithm, especially for complex specs. Comparison is made with  MILP-based optimal control~\cite{donzeBluSTLControllerSynthesis}, SMT-based model checking~\cite{YuLB22}, and optimization-based falsification~\cite{Donze10,ZhangLAMHZ21}.

    \item Through the algorithm, case studies and experiments, we argue for the importance and feasibility of \emph{spec analysis} for CPSs.

\end{itemize}
After exhibiting preliminaries on STL
and stable partitioning
in \cref{sec:prelim}, we formulate our problems (bounded trace synthesis, model checking, etc.) in \cref{sec:problemSetting}.
In \cref{sec:encodeSTL} we present a novel \emph{variable-interval MILP encoding} of STL;
in \cref{sec:systemModels} we discuss MILP encoding of a few families of models.
Our main algorithm  combines these two encodings. In \cref{sec:expr} we present experiment results.

\myparagraph{Related Work I: Optimal STL Control with MILP}
The works~\cite{ramanReactiveSynthesisSignal2015,RamanDMMSS14,donzeBluSTLControllerSynthesis} inspire our use of MILP for STL.  Their problem is \emph{optimal controller synthesis under STL constraints}, i.e.\ to find an input signal $\tau$ to a system model $\mathcal{M}$
so that 1) the output signal $\mathcal{M}(\tau)$ satisfies a given STL spec $\varphi$ and 2) it optimizes  $J(\mathcal{M}(\tau))$, where $J$ is a given objective function.  This problem subsumes our problem of trace synthesis, by taking a constant function as $J$.

The algorithms in~\cite{ramanReactiveSynthesisSignal2015,RamanDMMSS14,donzeBluSTLControllerSynthesis} reduce their problem to MILP by a \emph{constant-interval encoding} of the robust semantics~\cite{donzeRobustSatisfactionTemporal2010,fainekosRobustnessTemporalLogic2009} of STL (an enhanced encoding is presented in~\cite{KurtzL22}). Specifically, their system model  is discrete-time dynamics $x(t+\varDelta t)=f_{\mathrm{d}}(x(t), u(t), w(t))$ with a constant interval $\varDelta t$.

In contrast, in our \emph{variable-interval} encoding (\cref{sec:encodeSTL}), continuous time is discretized into the intervals $\dotsc, (\gamma_{i-1},\gamma_{i}),\{\gamma_{i}\}, (\gamma_{i},\gamma_{i+1}), \dotsc$ where the end points $\gamma_{i}$ are also variables in MILP. This is  advantageous not only in modeling precision but also in scalability: for system models that are largely continuous, constant-interval discretization incurs more integer variables in MILP, hampering the performance of MILP solvers. See \cref{sec:expr} for experimental comparison.

\myparagraph{Related Work II: SMT-Based STL Model Checking} Our key technical element (a variable-interval MILP encoding of STL) uses the idea of stable partitioning  from~\cite{baeBoundedModelChecking2019,leeEfficientSMTBasedModel2021,YuLB22}. They solve bounded STL model checking, and also its dual (trace synthesis). The main difference is the class of system models $\mathcal{M}$ accommodated. SMT solvers accommodate more theories than MILP solving, and thus allows encoding of a greater class of models. In contrast, by restricting the model class to MILP-encodable, our algorithm benefits speed and scalability (MILP is faster than SMT).   Iterative optimization in MILP also makes our algorithm an anytime one. Native support of parameter synthesis is another plus.

\myparagraph{Other Related Work} \emph{Optimization-based falsification} has its root in the quantitative robust semantics of STL~\cite{donzeRobustSatisfactionTemporal2010,fainekosRobustnessTemporalLogic2009}; the successful combination with stochastic optimization metaheuristics has made falsification an approach of both scientific and industrial interest. See the ARCH competition report~\cite{ErnstABCDFFG0KM21} for state-of-the-art.  Falsification is most of the time thought of as \emph{search-based testing}; therefore, unlike the model checking approach, the absence of counterexamples is usually not proved. Exceptions are~\cite{PedrielliKCTHCF23,ZhangA21} where they strive for probabilistic guarantees for such absence.

The current work is motivated by the observation that falsification solvers often struggle in trace synthesis for complex STL specs, even if a system model is simple. It is known that specs with more connectives pose a performance challenge, and many countermeasures are proposed, including~\cite{AkazakiH15} (for temporal operators) and~\cite{ZhangLAMHZ21,ZhangHA19} (for Boolean connectives).

%
%

\section{Preliminaries}\label{sec:prelim}
We let $\mathbb{N}, \mathbb{R}$ denote the sets of natural numbers and reals, respectively; $\mathbb{R}_{\ge 0}$ denotes an obvious subset. The set $\EReal=\Real\cup\{-\infty,\infty\}$ is that of extended reals. The set $\mathbb{B} = \{\top, \bot\}$ is for Boolean truth values.
The \emph{powerset} of a set $X$ is denoted by $\wp(X)$.
An \emph{interval} is a subset of $\mathbb{R}_{\ge 0}$ of the form $(a,b)$, $[a,b)$, $(a,b]$, or $[a,c]$, where $a<b$ and $a\le c$. Therefore a singleton $\{a\}$ is an interval.

\begin{definition}[linear predicate $p$ and $\sem{p}, \pi_{p}$]\label{def:linPred}
    Given a set $V$ of variables, a  \emph{(closed) linear predicate} is a function  $p\colon \mathbb{R}^V \to \mathbb{B}$ defined as follows, using  some $c \in \mathbb{R}^V$ and $b \in \mathbb{R}$:
    $p(x) = \top$ if and only if $c^{\top}x + b \geq 0$.
We write $\sem{p}$ for the closed half-space  $\{ x \mid p(x) = \top \} \subseteq \mathbb{R}^V$.

For the above $p$, we define a function $\pi_p(x)\colon \mathbb{R}^V \to \mathbb{R}$   by $\pi_{p}(x) \Defeq c^{\top}x + b$. This is understood as the degree of satisfaction (or violation, if negative) of a linear predicate $p$ by $x\in\mathbb{R}^V$. Indeed,   $\pi_{p}(x)$ is the (signed) Euclidean distance to the boundary of $\sem{p}$, assuming that the Euclidean norm of $c$ is $\lVert c \rVert = 1$.
\end{definition}

\begin{definition}[signal]\label{def:signal}
Let $V$ be a finite set of variables and $T$ a positive real.
A \emph{signal} over $V$ with a time horizon $T$ is a function $\sigma: [0, T] \to \mathbb{R}^V$.
We write $\Signal_V^T$ for the set of all signals over $V$ with time horizon $T$, or simply $\Signal_V$ when $T$ is clear from the context.

If necessary, the domain $[0, T]$ of $\signalSymbol$ can be extended to $\R_{\geq 0}$ by setting $\signalSymbol(t) \coloneqq \signalSymbol(T)$ for all $t > T$.
This allows us to define the notion of \emph{$t$-postfix}, which will serve as the basis of the STL semantics (\cref{sec:signalTemporalLogic}). Precisely,
the \emph{$t$-postfix} of $\sigma$ is a signal $\signalSymbol^t$
defined by $\signalSymbol^t(t') \Defeq \signalSymbol(t+t')$. The domain of $\signalSymbol^t$ can be chosen freely but we set it to $[0,T]$ for consistency.
\end{definition}

\begin{definition}[system model, trace set $\mathcal{L}(\mathcal{M})$]\label{def:sysModelTraceSet}
Let $V,V'$ be finite sets of variables.
A \emph{system model} $\mathcal{M}$ from  $V'$ to  $V$ with a time horizon $T$ is a function $\mathcal{M}\colon \Signal^T_{V'}\to\pow(\Signal_{V}^T)$.
The \emph{trace set} $\mathcal{L}(\mathcal{M})$ of a system model $\mathcal{M}$ is
\begin{math}
\mathcal{L}(\mathcal{M}) \Defeq \textstyle\bigcup_{\tau\in\Signal_{V'}^T}\mathcal{M}(\tau),
\end{math}
that is, the set of all output signals of $\mathcal{M}$ where an input signal $\tau$ can vary.
\end{definition}

\noindent
We allow system models to be nondeterministic (note the the powerset construction $\pow$); the models in \cref{sec:intro}  were deterministic for simplicity.
A special case of the above is when $V'=\emptyset$, that is, when $\mathcal{M}$ does not have any input. In this case,  $\Signal_{V'}$ is a singleton, and therefore a function $\mathcal{M}$ can be identified with a subset $\mathcal{L}(\mathcal{M})\subseteq\Signal_{V}$.


    \begin{example}[$\mathcal{M}_{\mathrm{RNC}}$]\label{ex:introLeadModel}
        The dynamics model in \cref{ex:introLead} is formalized as a system model $\MRNC$ whose input variables (in $V'$) are $a_{\front}, v^{\mathrm{init}}_{\front},x^{\mathrm{init}}_{\front},
            a_{\rear}, v^{\mathrm{init}}_{\rear},x^{\mathrm{init}}_{\rear}
        $, and output variables (in $V$) are  $a_{\front}, v_{\front},x_{\front},
            a_{\rear}, v_{\rear},x_{\rear}
        $.
        Here, the input is acceleration rates ($a_{\front},a_{\rear}$) and the initial values of velocities and positions (modeled using signals $v^{\mathrm{init}}_{\front}$ etc.\ for convenience).
The time horizon $T$ of $\mathcal{M}$ represents its simulation time; here we set $T=20$.
        Given an input signal $\tau$, the output  $\mathcal{M}(\tau)$ is a singleton $\mathcal{M}(\tau)=\{\sigma\}$, and $\sigma$ is determined by the ODE~\cref{eq:introLeadModel}. Specifically, $\sigma(t)(a_{\front})=\tau(t)(a_{\front})$, $\sigma(t)(v_{\front})=\tau(0)(v^{\mathrm{init}}_{\front})+\int_{0}^{t}\tau(t')(a_{\front})\,\mathrm{d}t'$, and so on.
    \end{example}

\subsection{Signal Temporal Logic}\label{sec:signalTemporalLogic}
    \begin{definition}[signal temporal logic (STL)]\label{def:stlSyntax}
In STL, an \emph{atomic proposition} over a variable set $V$ is represented as $p:\equiv (f(\vec{w}) \geq 0)$, where $f: \R^{\dimNum}\to\R$ is a function that maps a $\dimNum$-dimensional vector $\vec{w}$ to a real. The syntax of an STL formula $\varphi$ (over $V$) is defined as follows:
        \begin{math}
\varphi :\equiv p \mid \bot\mid
            \top\mid
            \neg\varphi \mid
            \varphi_1\lor \varphi_2\mid
            \varphi_1\wedge \varphi_2\mid
            \Diamond_I\varphi\mid
            \Box_I\varphi \mid
            \varphi_1 \UntilOp{I} \varphi_2 \mid \varphi_1 \ReleaseOp{I} \varphi_2
        \end{math},
        where $I$ is a nonsingular closed time interval, and $\Diamond_{I}$, $\Box_{I}, \UntilOp{I}$, $\ReleaseOp{I}$  are temporal operators \emph{eventually}, \emph{always}, \emph{until} and \emph{release}. Implication is defined: $\varphi_1 \Rightarrow \varphi_2 :\equiv \neg\varphi_1 \lor \varphi_2$.
We  write temporal operators without the subscript $I$ when $I=[0,\infty]$ (e.g., $\Diamond_{}$).
Note that we do not lose generality by restricting the inequality in  $p :\equiv (f(\vec{w}) \geq  0)$. Indeed, $\leq, <, >$ can be encoded using (a combination of) $-f$ and $\lnot$.

The set $\Sub(\varphi)$ collects all subformulas of an STL formula $\varphi$; the set $\AP(\varphi)$ collects all atomic propositions $\alpha$ occurring in $\varphi$.
    \end{definition}

\begin{proposition}\label{prop:NNF}
    Every STL formula has a formula in the \emph{negation normal form (NNF)}---i.e.\ one in which  negation $\lnot$ appears only in front of atomic propositions---that is semantically equivalent. \qed
\end{proposition}

    \begin{assumption}\label{asmp:onlyLinearPredicates}
We assume that each atomic proposition $p$  is a linear predicate (\cref{def:linPred}), that is, $f(x)=c^\top x+b$ with some $c\in\mathbb{R}^{V}, b\in\mathbb{R}$ in each $p :\equiv (f(\vec{w}) \geq  0)$.
    \end{assumption}

The Boolean semantics $\sigma\models \varphi$ and robust semantics $\sem{\sigma,\varphi}\in \EReal$ of STL are standard. See~\cref{sec:appendix}.

\emph{PSTL} is a parametric extension of STL.
It is from~\cite{AsarinDMN11}; see also~\cite{BartocciMNN22}.  Its definition is in~\cref{sec:appendix}.
The semantics of PSTL formula is defined naturally by fixing $\vec{u},\vec{v}$; see \cref{prob:existParMining} for the specific forms we use.

\subsection{Finite Variability}\label{sec:finiteVariability}
The satisfiability checking problem for STL---this is equivalent to the model checking problem under the trivial (identity) system model---is already  EXPSPACE-complete~\cite{alurBenefitsRelaxingPunctuality}.
To ease computational complexity, \emph{bounded model checking} has been a common approach~\cite{leeEfficientSMTBasedModel2021,prabhakarAutomaticTraceGeneration2018}.
Its main idea is to bound the number of time-points at which the truth value of each subformula can vary.

A (finite) \emph{partition} $\mathcal{P}$ of an interval $D \subseteq \mathbb{R}$ is a sequence  $\mathcal{P} = (J_{i})_{i = 1}^{N}$ of nonempty and mutually disjoint intervals such that $\bigcup_{i=1}^N J_{i} = D$, and $\sup(J_{i}) \leq \inf(J_{i'})$ for any $i < i'$.
    \begin{definition}[finite variability~\cite{rabinovichDecidabilityContinuousTime1998a}]\label{def:finiteVariability}
        A Boolean signal  $q\colon \mathbb{R}_{\geq 0} \to \mathbb{B} $ is
\emph{constant} on an interval $J  \subseteq \mathbb{R}_{\ge 0}$ if $q(t) = q(t')$  for any $t, t' \in J$.  
We say  $q(t)$ has \emph{$N$-bounded variability} if there exists a partition $\mathcal{P}$ of $[0, \infty)$ and $q(t)$ is constant on every interval $J \in \mathcal{P}$.

Let $\sigma\colon [0, T] \to \mathbb{R}^V$ be a signal and $\varphi$ be an STL formula over $V$. We say that $\sigma$ has the \emph{$N$-bounded variability} with respect to $\varphi$ if the Boolean (truth value) signal $t \mapsto (\sigma^t \models \varphi)$ has the  $N$-bounded variability.
We say $\sigma$ is \emph{finitely variable} with respect to $\varphi$ if it has the $N$-bounded variability for some $N$.

        Finally, we say $\sigma$ has the \emph{hereditary $N$-bounded variability} with respect to $\varphi$ if, for each subformula $\psi\in\Sub(\varphi)$, $\sigma$ has the $N$-bounded variability with respect to $\psi$. We write \emph{$N$-HBV} for the hereditary $N$-bounded variability.
\end{definition}

    \begin{lemma}[\!\!{\cite{baeBoundedModelChecking2019}}]\label{lem:finiteVariabilitySTL}
        Let $\varphi$ be an STL formula.
A signal $\sigma$ has the $N$-HBV with respect to $\varphi$ for some $N$ if and only if it is finitely variable with respect to each atomic proposition $p\in\AP(\varphi)$ occurring in $\varphi$.
        \qed
    \end{lemma}

    The following is the basis of bounded model checking in~\cite{baeBoundedModelChecking2019,leeEfficientSMTBasedModel2021}.
    \begin{definition}[stable partition]\label{def:stablePartitioning}
        Let $\sigma$ be a signal, $\varphi$ be an STL formula, and $\mathcal{P}$ be a partition of $[0, T]$ such that every $J \in \mathcal{P}$ is singular or open.
        Intuitively, $\mathcal{P}$ looks like $\{\gamma_0\},(\gamma_{0},\gamma_{1}),\{\gamma_{1}\},(\gamma_{1},\gamma_{2}),\dotsc, \{\gamma_N\}$.
        We say $\mathcal{P}$ is a \emph{stable partition} for $\sigma$ and $\varphi$ if $t \mapsto \sigma^t \models \psi$ is constant on $J$ for each  $J \in \mathcal{P}$ and  $\psi \in \Sub(\varphi)$.
    \end{definition}


\section{Problem Formulation}\label{sec:problemSetting}
We formulate our problems and discuss their mutual relationship.
The next problem is studied in~\cite{baeBoundedModelChecking2019,leeEfficientSMTBasedModel2021,YuLB22}.

\begin{problem}[bounded STL model checking]\label{prob:boundedModelChecking}
Given an STL formula $\varphi$ (over $V$), a system model $\mathcal{M}$ (from $V'$ to $V$) with time horizon $T$, and a variability bound $N\in\mathbb{N}$, decide if the following is true or not:
$\sigma \models \varphi$ holds
for an arbitrary trace $\sigma\in\mathcal{L}(\mathcal{M})$ (cf.\ \cref{def:sysModelTraceSet}) that has the
hereditary $N$-bounded variability
($N$-HBV)
with respect to $\varphi$.
%
\end{problem}


The following  is the dual of \cref{prob:boundedModelChecking}, and is our main scope.

    \begin{problem}[bounded STL trace synthesis]\label{prob:traceSynthesis}
    Given $\varphi,\mathcal{M},T$ and $N$  as in \cref{prob:boundedModelChecking}, find a trace
    $\sigma \in \mathcal{L}(\mathcal{M})$  such that 1) $\sigma$ has the
    $N$-HBV with respect to $\varphi$
and 2)  $\sigma\models \varphi$ holds, or prove that such $\sigma$ does not exist.
    \end{problem}




\cref{prob:traceSynthesis} resembles
the \emph{falsification problem}~\cite{fainekosRobustnessTemporalLogic2009}: given  $\mathcal{M}$ (that can be black-box) and $\varphi'$, find a \emph{counterexample input} $\tau$ such that $\mathcal{M}(\tau)\not\models \varphi'$.  The emphases and the settings are often different though; see \cref{sec:intro}.


\noindent
The following is a special case of the \emph{STL parameter mining} problem; see e.g.~\cite[\S{}3.5]{BartocciMNN22}.
Recall from~\cite[Def. A.3]{satoCAV24Extended} that  $\varphi_{\vec{u},\vec{v}}$ instantiates parameters $\vec{p},\vec{q}$ in $\varphi$ with real values $\vec{u},\vec{v}$ from the domains $P,Q$, respectively.

    \begin{problem}[bounded existential parameter mining]\label{prob:existParMining}
    Let $\varphi$ be a PSTL formula over  parameters $(\vec{p},\vec{q})$, and  $\mathcal{M},T$ and $N$  be as in \cref{prob:boundedModelChecking}. Find the set
    \begin{math}
\bigl\{\,(\vec{u},\vec{v})\in P\times Q\,\big|\, \sigma\models\varphi_{\vec{u},\vec{v}} \text{ for some $\sigma\in\mathcal{L}(\mathcal{M})$ that has the
            $N$-HBV wrt.\ $\varphi$
        }\,\bigr\}.
    \end{math}
    \end{problem}
\noindent    In \cref{sec:expr}, we study a further special case where there is only one parameter $p$ and the goal is to find the maximum  $p$ in the above set.

\section{Variable-Interval Encoding of STL to MILP}\label{sec:encodeSTL}
\begin{auxproof}
    We present our novel variable-interval encoding of STL to MILP. Our optimization-based trace synthesis/model checking algorithm hinges on it.
\end{auxproof}




\subsection{$\delta$-Stable Partitions}\label{sec:deltaStablePartitioning}
We shall adapt the idea of stable partitioning~\cite{baeBoundedModelChecking2019}, reviewed in~\cref{sec:finiteVariability}, to the current MILP setting. A major difference we need to address is that SMT is symbolic while MILP is numerical: most MILP solvers
do not distinguish $<$ from $\le$ and do not accommodate strict inequalities.
See e.g.~\cite{gurobi}.

In order to address this difference,
we  develop a theory
of \emph{$\delta$-stable partitions}. Here is its outline.
Firstly, we replace partitions $\dotsc,(\gamma_{i-1},\gamma_{i}),\{\gamma_{i}\},(\gamma_{i},\gamma_{i+1}), \dotsc$  used in~\cite{baeBoundedModelChecking2019} (see also \cref{def:stablePartitioning}) with new ``partitions''  $\dotsc,[\gamma_{i-1},\gamma_{i}],[\gamma_{i},\gamma_{i+1}], \dotsc$. The latter can be expressed only using $\le$; but they have overlaps (at $\gamma_{i}$).
The original stability notion (see~\cref{sec:finiteVariability}) does not fit the new partition notion---it requires ``constantly true'' or ``never true,'' and prohibits overlaps. Therefore we introduce  \emph{$\delta$-stability}; it requires either ``constantly true'' or ``never \emph{robustly} true.''

\begin{figure}[tbp]
 \begin{minipage}[t]{.48\textwidth}
        \centering

        \scalebox{.7}{
            \begin{tikzpicture}[baseline=(current bounding box.north)]
 \coordinate[label=below left:O] (O) at (0,0);
 \draw[thin,-latex](O) --(7.3,0) node[right] {$t$};
 \draw[thin,-latex](O) --(0,2) node[left] {$x$};
 \draw[thick] (0,1.4)--(1,1)--(1.6,0.2)--(3,1)--(4.2,1.6)--(4.5,1.3)--(5,1.6)--(6,1)--(7,0.4);
 \coordinate (T1) at (0, -1);
 \coordinate (T2) at (1, -1);
 \coordinate (T3) at (3, -1);
 \coordinate (T4) at (6, -1);
 \coordinate (T5) at (7, -1);
 \draw [dashed,thin] (0,1) -- (7,1) node[anchor=west]{$x \geq 1$};
 %
 \draw[thick,-latex](-0.3,-1) --(7.3,-1);
 %
 \draw [dotted,thin] ([shift={(0,3)}]T1) -- (T1) node[below=0.1]{$J_1$};
                \draw[(-),thick] (T1) -- (T2) node[midway,below=0.1]{$J_2$};
 \draw [dotted,thin] ([shift={(0,3)}]T2) -- (T2) node[below=0.1]{$J_3$};
 \draw[(-),thick] (T2) -- (T3) node[midway,below=0.1]{$J_4$};
 \draw [dotted,thin] ([shift={(0,3)}]T3) -- (T3) node[below=0.1]{$J_5$};
 \draw[(-),thick] (T3) -- (T4) node[midway,below=0.1]{$J_6$};
 \draw [dotted,thin] ([shift={(0,3)}]T4) -- (T4) node[below=0.1]{$J_7$};
 \draw[(-),thick] (T4) -- (T5) node[midway,below=0.1]{$J_8$};
 \draw [dotted,thin] ([shift={(0,3)}]T5) -- (T5) node[below=0.1]{$J_9$};
 \node[below=0.7] () at (-0.5,-1) {$\varphi$};
 \node[below=0.7] (J1) at (T1) {$\top$};
 \node[below=0.7] (J2) at ($(T1)!0.5!(T2)$) {$\top$};
 \node[below=0.7] (J3) at (T2) {$\top$};
 \node[below=0.7] (J4) at ($(T2)!0.5!(T3)$) {$\bot$};
 \node[below=0.7] (J5) at (T3) {$\top$};
 \node[below=0.7] (J6) at ($(T3)!0.5!(T4)$) {$\top$};
 \node[below=0.7] (J7) at (T4) {$\top$};
 \node[below=0.7] (J8) at ($(T4)!0.5!(T5)$) {$\bot$};
 \node[below=0.7] (J9) at (T5) {$\bot$};
            \end{tikzpicture}
        }

        \caption{A stable partition (cf.~\cite{baeBoundedModelChecking2019}) for $\sigma$ and $\varphi :\equiv x \geq 1$. The symbols $\top$ and $\bot$ denote the (constant) truth value of $\varphi$ each interval $J_i$.}
        \label{fig:stablePartitionSeq}
 \end{minipage}
 \hfill
 \begin{minipage}[t]{.48\textwidth}
        \centering

        \scalebox{.7}{
            \begin{tikzpicture}[baseline=(current bounding box.north)]
 \coordinate[label=below left:O] (O) at (0,0);

 \coordinate (T1) at (0, -1);
 \coordinate (T2) at (1, -1);
 \coordinate (T2') at (0.5, -1);
 \coordinate (T2'') at (0.8, -1);
 \coordinate (T3) at (3, -1);
 \coordinate (T3') at (3.3, -1);
 \coordinate (T3'') at (3.1, -1);
 \coordinate (T4) at (6, -1);
 \coordinate (T4') at (5.7, -1);
 \coordinate (T4'') at (5.8, -1);
 \coordinate (T5) at (7, -1);

 \fill[blue!20] ([shift={(0,3)}]T2) rectangle (T2');
 \fill[blue!20] ([shift={(0,3)}]T3) rectangle (T3');
 \fill[blue!20] ([shift={(0,3)}]T4) rectangle (T4');

\draw[thin,-latex](O) --(7.3,0) node[right] {$t$};
\draw[thin,-latex](O) --(0,2) node[left] {$x$};
\draw[thick] (0,1.4)--(1,1)--(1.6,0.2)--(3,1)--(4.2,1.6)--(4.5,1.3)--(5,1.6)--(6,1)--(7,0.4);

 \draw [dashed,thin] (0,1) -- (7,1) node[anchor=west]{$x \geq 1$};
 \draw [dashed,thin] (0,1.2) -- (7,1.2) node[anchor=south west]{$x \geq 1 + \delta$};
 %
 \draw[thick,-latex](-0.3,-1) --(7.3,-1);
 %
 \draw [dotted,thin] ([shift={(0,3)}]T1) -- (T1) ;
                \draw[{[-]},thick] (T1) -- (T2'') node[midway,below=0.1]{$\Gamma_1$};
 \draw [dotted,thin] (T2) -- (T2) ;
 \draw [dotted,thin] ([shift={(0,3)}]T2') -- (T2') ;
 \draw[{[-]},thick] (T2'') -- (T3'') node[midway,below=0.1]{$\Gamma_2$};
 \draw [dotted,thin] ([shift={(0,3)}]T3) -- (T3) ;
 \draw [dotted,thin] ([shift={(0,3)}]T3') -- (T3') ;
 \draw[{[-]},thick] (T3'') -- (T4'') node[midway,below=0.1]{$\Gamma_3$};
 \draw [dotted,thin] ([shift={(0,3)}]T4) -- (T4) ;
 \draw [dotted,thin] ([shift={(0,3)}]T4') -- (T4') ;
 \draw[{[-]},thick] (T4'') -- (T5) node[midway,below=0.1]{$\Gamma_4$};
 \draw [dotted,thin] ([shift={(0,3)}]T5) -- (T5) ;
 \node[below=0.7] () at (-0.5,-1) {$\varphi\phantom{\delta}$};
 \node[below=0.7] (J2) at ($(T1)!0.5!(T2'')$) {$\top$};
 \node[below=0.7] (J4) at ($(T2'')!0.5!(T3'')$) {$?$};
 \node[below=0.7] (J6) at ($(T3'')!0.5!(T4'')$) {$\top$};
 \node[below=0.7] (J8) at ($(T4'')!0.5!(T5)$) {$?$};
 \node[below=1.2] () at (-0.5,-1) {$\varphi^\delta$};
 \node[below=1.2] (J2) at ($(T1)!0.5!(T2'')$) {$?$};
 \node[below=1.2] (J4) at ($(T2'')!0.5!(T3'')$) {$\bot$};
 \node[below=1.2] (J6) at ($(T3'')!0.5!(T4'')$) {$?$};
 \node[below=1.2] (J8) at ($(T4'')!0.5!(T5)$) {$\bot$};
            \end{tikzpicture}
        }

 \caption{A $\delta$-stable partition (\cref{def:deltaStable}) for $\sigma$ and $\varphi
$.
Here $\varphi^{\delta}\equiv (x\ge 1+\delta)$.
$\top$ and $\bot$ are much like in \cref{fig:stablePartitionSeq};
the symbol $?$ indicates that the truth value is not constant in that interval.
 In some regions (shaded), $\sigma^t \models \varphi$ is true but     $\sigma^t \models \varphi^\delta$ is not.
 }
        \label{fig:deltaStablePartitionSeq}
 \end{minipage}
\end{figure}

\begin{example}
    Let $\sigma$ be a continuous signal.
    Suppose that a sequence $\mathcal{P} = {(J_i)}_{i=1}^M$ is a stable partition for $\sigma$ and an STL formula $\varphi$, as illustrated in \cref{fig:stablePartitionSeq}.
By definition, intervals $J_i$ are mutually disjoint and constitute an alternating sequence of open intervals and singular intervals. In \cref{fig:stablePartitionSeq}, the time domain $[0,T]$ is split into nine intervals.


    In this paper, since MILP solvers do not accommodate strict inequalities, we are forced to use closed intervals; see $\Gamma_{1},\dotsc, \Gamma_{4}$ in \cref{fig:deltaStablePartitionSeq}. Notice that the truth value of the formula $\varphi$ is not constant in $\Gamma_{2}$ or $\Gamma_{4}$. To regain stability, we introduce the \emph{$\delta$-tightening} $\varphi^{\delta}$ of the formula $\varphi$ with some $\delta > 0$ (\cref{def:deltaTightening}); here $\varphi^{\delta}\equiv (x\ge 1+\delta)$. Then the truth value of $\varphi^{\delta}$ (instead of $\varphi$) is constantly false in $\Gamma_{2}$ and $\Gamma_{4}$, that is, $\varphi$ is ``never $\delta$-robustly true'' in $\Gamma_{2}$ and $\Gamma_{4}$.

\end{example}

\begin{definition}[time sequence, timed state sequence]\label{def:timedStateSequence}
    A \emph{time sequence}
    of $[0, T]$ is a sequence $\Gamma=(0 = \gamma_0 <
\dots < \gamma_N = T)$. Such a  time sequence induces a ``partition of $[0,T]$ with singular overlaps,'' namely $\Gamma=\bigl([\gamma_{i-1}, \gamma_{i}]\bigr)_{i=1}^N$. We identify it with the original time sequence,
     writing $\Gamma_i$ for the interval $[\gamma_{i-1}, \gamma_{i}]$.

    Given a  time sequence,
    a \emph{timed state sequence} over $V$ is a sequence $\varsigma = \bigl((x_0, \gamma_0), \dots, (x_N, \gamma_N)\bigr)$,
    where $x_0, \dots, x_N$ in $\mathbb{R}^V$.
\end{definition}
In MILP, it is efficient to represent signals as (continuous) \emph{piecewise-linear signals}, so that values within an interval can be deduced by linear interpolation.

\begin{definition}[piecewise-linear signal]\label{def:piecewiseLinearSignal}
    Given a timed state sequence $\varsigma = ((x_0, \gamma_0), \dots, (x_N, \gamma_N))$,
    the signal $\varsigma^\mathrm{pwl}\colon [0, \gamma_N] \to \mathbb{R}^V$ is defined by the following linear interpolation: $\varsigma^\mathrm{pwl}(t) \Defeq (1- \lambda) x_{i-1} + \lambda x_{i}$ if $\gamma_{i-1} \leq t \leq \gamma_i$ (where  $\lambda = \frac{1}{\gamma_{i} - \gamma_{i-1}} (t - \gamma_{i-1})$).

    In this paper, a \emph{piecewise-linear signal} is a signal of the form $\varsigma^\mathrm{pwl}$ for some timed state sequence $\varsigma$. Note that it is continuous everywhere, and is linear everywhere  except for only finitely many points.
\end{definition}
Obviously, $\varsigma^\mathrm{pwl}$ is finitely variable with respect to any linear predicate $p$ (\cref{def:linPred}).

\begin{definition}[$\delta$-tightening of linear predicates]\label{def:deltaTightening}
    Let $\delta > 0$ be a positive real and $p$ be a linear predicate defined by $p(x) = \top \iff c^{\top}x + b \geq 0$.
    The \emph{$\delta$-tightening} of $p$ is a linear predicate defined by $p^\delta(x) = \top \iff c^{\top}x + b \geq \delta$.
\end{definition}

\noindent
Note that $p^\delta$ is stronger than $p$, i.e., $[\![p^\delta]\!] \subsetneq [\![ p ]\!]$.
We further extend the concept of $\delta$-tightening for general STL formulas in NNF (cf.\ \cref{prop:NNF}).
Let $p^-$ be the linear predicate defined by  $p^-(x) = \top \iff -c^{\top}x - b \geq 0$.
\begin{definition}[$\delta$-tightening of STL formulas in NNF]
    Let $\varphi$ be an STL formula in NNF.
    The \emph{$\delta$-tightening}  $\varphi^\delta$ of $\varphi$
    is the STL formula obtained from $\varphi$ by replacing all occurrences of atomic predicates $p$ (resp.\ $\lnot p$)  by $p^\delta$ (resp. $(p^-)^\delta$).
\end{definition}

\noindent
The $\delta$-tightening construction is related to robust semantics (\cref{def:STLRobustSem}).
\begin{proposition}\label{prop:nnfSemantics}
    Let $\sigma$ be a signal, $\varphi$ be an STL formula in NNF, and  $\delta >0$.
    Then $\sigma \models \varphi^\delta$ implies $[\![\sigma, \varphi]\!] \geq \delta$.
\qed
\end{proposition}

\noindent
Since the closed halfspace $[\![p^-]\!]$ coincides with the closure of the open halfspace $\mathbb{R}^V \setminus [\![p]\!]$, the robust semantics is not affected by the difference between $p^-$ and $\lnot p$.
For simplicity, in the following, we assume that any STL formula in NNF does not contain negation, i.e., $\lnot p$ is replaced by a new atomic proposition $p^-$.

We are ready to define \emph{$\delta$-stability}.
\begin{definition}[$\delta$-stability]\label{def:deltaStable}
Let $\varphi$  be an STL formula over $V$ in NNF, $\sigma\in\Signal_{V}^T$ be a signal,
and  $\Gamma=(\gamma_0, \dots, \gamma_N)$ be a  time sequence (\cref{def:timedStateSequence}) with $\gamma_{N}=T$.
We say $\Gamma$ is \emph{$\delta$-stable} for $\sigma$ and $\varphi$ if, for each $i \in [1,N]$ and each subformula $\psi\in\Sub(\varphi)$, either of the following holds: 1) $\sigma^t \models \psi$  for each $t\in \Gamma_i$, or 2) $\sigma^t \not\models \psi^\delta$   for each $t\in\Gamma_i$.
\end{definition}
In the above definition, in each interval $\Gamma_{i}$, a subformula $\psi$ is either 1) always true or 2) never robustly true.  The two conditions are not mutually exclusive---both hold if $\sigma^{t}\models \psi\land\lnot \psi^{\delta}$ for all $t\in\Gamma_{i}$. 

The next notion of conservative valuation records which of 1) and 2) is true in each interval.  It conservatively approximates the actual truth of $\varphi$ (\cref{fig:deltaStablePartitionSeq}).
\begin{definition}[conservative valuation]\label{def:conservValuation}
Let $\varphi$ be an STL formula in NNF, and $\Gamma=(\gamma_0, \dots, \gamma_N)$  be a  time sequence . A \emph{valuation} of $\varphi$ in $\Gamma$ is a function  $\Theta: \mathrm{Sub}(\varphi) \times [1,N] \to \mathbb{B}$ that assigns, to each subformula and  index of the intervals of $\Gamma$, a Boolean truth value.
Let $\sigma$ be a signal with a time horizon $T = \gamma_N$. We say that $\Theta$ is a \emph{conservative valuation} of $\varphi$ in $\Gamma$ on $\sigma$  (up to $\delta$) if
    1) $\Theta(\psi, i) = \top$ implies that, for each $t\in\Gamma_{i}$, $\sigma^t \models \psi$ holds; and 2) $\Theta(\psi, \Gamma_{i}) = \bot$ implies, for each $t\in\Gamma_i$, $\sigma^t \not\models \psi^\delta$.
\end{definition}
We simply write $\langle \psi \rangle_i$ for $\Theta(\psi, i)$ when $\Theta$ is clear from context.

\begin{lemma}
Suppose there exists a conservative valuation $\Theta$ of an STL formula $\varphi$ in a time sequence $\Gamma$ on a signal $\sigma$ up to $\delta$.
Then $\Gamma$ is $\delta$-stable for $\sigma$ and $\varphi$. \qed
\end{lemma}

\begin{example}
In \cref{fig:deltaStablePartitionSeq}, we have a conservative valuation $\Theta$ for the formula $\varphi \equiv x \geq 1$ such that $\Theta(1, \varphi) = \top$, $\Theta(2, \varphi) = \bot$, $\Theta(3, \varphi) = \top$, $\Theta(1, \varphi) = \bot$.
\end{example}

We shall argue in  \cref{sec:encodeFullSTL} that, for each piecewise-linear signal $\sigma$ (\cref{def:piecewiseLinearSignal}), an STL formula $\varphi$,
there is a time sequence $\Gamma$ in which $\varphi$ is $\delta$-stable on $\sigma$. We start with a special case where $\varphi$ is an atomic proposition $p$.

\begin{wrapfigure}[12]{r}[0pt]{4.3cm}
    \vspace{-1.5em}
    \centering
    \scalebox{.6}{    \begin{tikzpicture}[scale=1]

            \tikzset{
                pwl/.style={circle,fill, scale=0.3},
            }

            \fill[gray!30](0,0)--(6,0)--(6,6);
            \draw[thick,dashed](1,0)--(6,5);
            \draw[thick](0,0)--(6,6) node[above,rotate=45, xshift=-0.6cm]{$\pi_p(x) = 0$};
            \draw[<->,dashed](5,5)--(5.5,4.5) node[above, xshift=-0.1cm, yshift=0.2cm] {$\delta$};

            \node [pwl,label=$\sigma(\gamma_0)$] (G0) at (3,5.5) {};
            \node [pwl,label={[xshift=0.3cm]$\sigma(\gamma_1)$}] (G1) at (4.5,4) {};
            \node [pwl,label={[xshift=0.3cm]$\sigma(\gamma_2)$}] (G2) at (5,3) {};
            \node [pwl,label={[yshift=-0.6cm]$\sigma(\gamma_3)$}] (G3) at (4,2) {};
            \node [pwl,label={[xshift=0.25cm]$\sigma(\gamma_4)$}] (G4) at (2,1.5) {};
            \node [pwl,label=$\sigma(\gamma_5)$] (G5) at (0,1) {};

            \draw[->](G0)--(G1) {};
            \draw[->,red,thick](G1)--(G2) {};
            \draw[->,red,thick](G2)--(G3) {};
            \draw[->,red,thick](G3)--(G4) {};
            \draw[->](G4)--(G5) {};

            \node at (5,1) {$\sem{p}$};
        \end{tikzpicture}
    }    \caption{The $\delta$-crossing pairs $(\sigma(\gamma_1), \sigma(\gamma_2))$, $(\sigma(\gamma_3), \sigma(\gamma_4))$ are stationary. The red segments are assigned $\top$ by a conservative valuation.
    }
    \label{fig:deltaCrossingPair}
\end{wrapfigure}
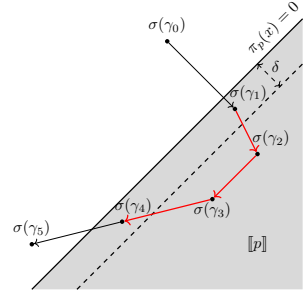

%
\begin{definition}\label{def:crossingPair}
    Let $x, x' \in \mathbb{R}^V$, and $p$ be a linear predicate on $V$. We say $(x, x')$ is a \emph{$\delta$-crossing pair} with respect to $p$ if $x\in \sem{p^\delta}$ and $x'\not\in\sem{p^\delta}$ (cf.\ \cref{def:linPred}), or vice versa.
    A $\delta$-crossing pair is \emph{stationary} if $x\in\sem{p}$ and $x'\in\sem{p}$.
\end{definition}

\begin{lemma}\label{lem:existenceOfDeltaStablePartitioning}
    Let $p$ be a  linear predicate
    and $\sigma$ be a  piecewise-linear signal.
    There is a  time sequence $\Gamma=(\gamma_{0}, \dots,\gamma_{N})$ such that, for any $i \in [1,N]$,
    1) $\sigma$ is a linear function on the interval $[\gamma_{i-1}, \gamma_{i}]$, and
2) if $(\sigma(\gamma_{i-1}), \sigma(\gamma_{{i}}))$ is a $\delta$-crossing pair, it is stationary. It follows that there is a conservative valuation $\Theta$ of $p$ in $\Gamma$ on $\sigma$.
\end{lemma}

\begin{proof}
    The lemma argues that, whenever $\sigma$ enters or leaves $\sem{p^{\delta}}$, it has to do so via $\sem{p}\setminus \sem{p^{\delta}}$. See \cref{fig:deltaCrossingPair}. This can be enforced by adding suitable  points to $\Gamma$, exploiting  continuity of $\sigma$  (\cref{def:piecewiseLinearSignal}) and the intermediate value theorem.   \qed
\end{proof}
%

\begin{auxproof}
    \begin{lemma}\label{prop:deltaStableConditionAtomic}
        Let $p, \sigma, \Gamma$ be as in~\cref{lem:existenceOfDeltaStablePartitioning}. Then
        $\Gamma$ is $\delta$-stable for $\sigma$ and $p$. \qed
    \end{lemma}

    \begin{proof}
        The goal is to show that, if both of $\sigma(\gamma_{i-1}) \in [\![p]\!]$ and $\sigma(\gamma_{i}) \in [\![p]\!]$ holds, then $\sigma(t) \in [\![p]\!]$;
        otherwise, $\sigma(t) \not\in [\![p^\delta]\!]$.
        It is clear that, we only need to check the case either $\sigma(\gamma_{i-1}) \not\in [\![p]\!]$ or $\sigma(\gamma_{i}) \not\in [\![p]\!]$ holds, but not both.
        Suppose $\sigma(\gamma_{i-1}) \not\in [\![p]\!]$ and $\sigma(\gamma_{i}) \in [\![p]\!]$.
        Since it results in a $\delta$-crossing pair that is not stationary, $\sigma(\gamma_{i-1})$ must not be in $[\![p^\delta]\!]$.
        We conclude $\sigma(\gamma_{i-1}), \sigma(\gamma_{i}) \not \in [\![p^\delta]\!]$, and consequently $\sigma(t) \not \in [\![p^\delta]\!]$.
        The opposite case can be checked in the same way. \qed
    \end{proof}
\end{auxproof}


We note another advantage of $\delta$-stable partitions: the number of invervals is roughly halved compared to (original) stable partitions (see \cref{fig:stablePartitionSeq,fig:deltaStablePartitionSeq}). This advantage may be exploited also in SMT-based approaches~\cite{baeBoundedModelChecking2019} for scalability.





\subsection{Variable-Interval MILP Encoding}\label{sec:encodeFullSTL}
Our  MILP encoding of STL relies on the constructs in \cref{sec:deltaStablePartitioning}. For the purpose of trace synthesis for an STL formula $\varphi$, our basic strategy is
to search for
1)  a time sequence  $\Gamma=(\gamma_0, \dots, \gamma_N)$ (i.e.\ a ``partition,'' see \cref{def:timedStateSequence}) and
2) a valuation $\Theta: \mathrm{Sub}(\varphi) \times [1,N] \to \mathbb{B}$,
such that
\begin{itemize}
\item $\Theta$ is \emph{consistent} in the sense that the truth values assigned to subformulas comply with the STL semantics (\cref{sec:signalTemporalLogic});
    \item $\Theta$ is \emph{fulfilling} in the sense that it assigns $\top$ to the top-level formula $\varphi$ in $\Gamma_{1}$ (the first interval); and
    \item $\Theta$ is \emph{realizable} in the sense that there is a piecewise-linear trace  $\sigma\in\mathcal{L}(\mathcal{M})$ of $\mathcal{M}$ that \emph{yields} $\Theta$. That is, precisely,   $\Theta$ must be a conservative valuation of $\varphi$ in $\Gamma$ on $\sigma$ (\cref{def:conservValuation}).
\end{itemize}
The entities $\Gamma,\Theta$ we search for are expressed as MILP variables, and the above three conditions  are expressed as MILP constraints. We describe these MILP variables and constraints in the rest of the section. The constraints expressing $\sigma\in\mathcal{L}(\mathcal{M})$ require system model encoding and are thus deferred to later sections.

\myparagraph{Variables} We use the following MILP variables. Their collection is denoted by $\mathbf{Var}(\varphi, N)$. Here $N\in \mathbb{N}$ is a constant for variability bound (\cref{prob:traceSynthesis}).
\begin{itemize}
    \item Real-valued variables $\{ \gamma_{0}, \dots ,\gamma_{N} \}$ for a  time sequence $\Gamma$.
    \item  Boolean variables $\{ \langle \psi \rangle_{i} \mid 1 \leq i \leq N, \psi \in \Sub(\varphi ) \}$ for the value $\Theta(\psi,i)$ of a valuation $\Theta$ that we search for.
    \item  Real-valued variables
          $\{ x_{i,v}\mid 0\leq i\leq N , v\in V\}$ for the values of a piecewise-linear trace $\sigma\in\mathcal{L}(\mathcal{M})$.
    \item  Boolean variables $\{ \zeta^p_{i}, \zeta^{\delta, p}_{i} \mid 0 \leq i \leq N, p \in \AP(\varphi ) \}$ for the truth values of $p$ and $p^\delta$ at time $\gamma_{i}$. These auxiliary variables are used to detect crossing pairs (\cref{def:crossingPair}).

    \item Real-valued variables $\{ S^\psi_{i} \mid 0 \leq i \leq N, \Box_{I}\psi \in \Sub(\varphi) \}$. This auxiliary variable records for how long $\psi$ has been true before $\gamma_{i}$.   

    \item Real-valued variables $\{ P^\psi_{i} \mid 0 \leq i \leq N, \Diamond_{I}\psi \in \Sub(\varphi) \}$. This auxiliary variable records for how long  $\psi$ has been false before $\gamma_{i}$. 

\end{itemize}

By an \emph{assignment} we refer to a function $\boldsymbol{v}\colon \mathbf{Var}(\varphi, N)\to \Real$ such that $\boldsymbol{v}(y)\in\{0,1\}$ for each Boolean variable $y$.  The MILP problem is to find an assignment $\boldsymbol{v}$ that optimizes an objective under given constraints.

%
%
%
%
%
%

\begin{notation}
    In what follows, as a notational convention, we simply write a variable $y$ for the value $\boldsymbol{v}(y)$ when the assignment $\boldsymbol{v}$ is clear from context.

    We write   $\varsigma$ for the timed state sequence composed of the time sequence $\{\gamma_{0}, \dots ,\gamma_{N}\}$ and the trace values  $\{ x_{j,v}\mid 0\leq j\leq N , v\in V\}$.
\end{notation}

Note that, in this paper, we encode the \emph{Boolean} semantics of STL~\cite[Def. A.1]{satoCAV24Extended}, unlike~\cite{ramanReactiveSynthesisSignal2015,RamanDMMSS14} where the \emph{robust} semantics~\cite[Def. A.2]{satoCAV24Extended} is encoded in a constant-interval manner. The combination of variable-interval encoding and quantitative robust semantics is future work; for example, a quantitative extension of $\delta$-stable partitions (\cref{sec:deltaStablePartitioning}) seems quite nontrivial.

\myparagraph{Shorthands for Propositional Connectives%
}
We use standard  shorthands for Boolean connectives in MILP constraints  (such as $\lnot A, A\land B$ where $A,B$ are Boolean variables). They are defined in \cref{sec:shorthand}.

\myparagraph{Realizability Constraints: Traces and Atomic Propositions}
We need to constrain $\gamma_0,\dots,\gamma_N$ to be a time sequence (\cref{def:timedStateSequence}), using some constant $\varepsilon>0$ and letting $\cdots\ge \varepsilon$ stand for  $\cdots>0$.
\begin{equation}\label{eq:timeSequenceConstraint}
    \begin{array}{ll}
        \gamma_0 = 0,\quad \gamma_N = T,\quad \gamma_{i} - \gamma_{i-1} \geq \varepsilon
         & \quad\text{for all $i \in [1,N]$}
    \end{array}
\end{equation}
For each $i$ and  $p\in \mathrm{AP}(\varphi)$ (say  $p$ is defined by $c^{\top}x + b \geq 0$), the variables $\zeta^p_i, \zeta^{\delta,p}_i$ are constrained as follows,

\begin{equation}\label{eq:AtomicLinearConstraint1}
    \begin{array}{ll}
        \zeta^p_i = 1 \;\Rightarrow\; c^{\top}x_i + b \geq 0\qquad\qquad\qquad
         & \zeta^p_i = 0 \;\Rightarrow\; c^{\top}x_i + b \leq -\varepsilon                 \\
        \zeta^{\delta, p}_i = 1 \;\Rightarrow\; c^{\top}x_i + b \geq \delta
         & \zeta^{\delta, p}_i = 0 \;\Rightarrow\; c^{\top}x_i + b \leq \delta-\varepsilon
    \end{array}
\end{equation}
Moreover, we impose the following to ensure that $\Gamma$ is the one in \cref{lem:existenceOfDeltaStablePartitioning}:
\begin{equation}\label{eq:AtomicLinearConstraint2}
    \begin{array}{ll}
        \zeta^{\delta, p}_i = 0 \land \zeta^{\delta, p}_{i+1} = 1 \;\Rightarrow\; \zeta^p_i = 1
        , \quad
        \zeta^{\delta, p}_i = 1 \land \zeta^{\delta, p}_{i+1} = 0 \;\Rightarrow\; \zeta^p_{i+1} = 1
    \end{array}
\end{equation}
Under constraints~\cref{eq:timeSequenceConstraint,eq:AtomicLinearConstraint1,eq:AtomicLinearConstraint2},
$\Gamma$ is $\delta$-stable for $\varsigma^\mathrm{pwl}$ (cf.\ \cref{def:piecewiseLinearSignal}) and $p$, by \cref{lem:existenceOfDeltaStablePartitioning}. By the definition of $\delta$-stability, we can now constrain the variable $\langle p \rangle_i$ by
$\langle p \rangle_i = \zeta^{\delta,p}_{i-1} \lor \zeta^{\delta,p}_{i}$
for each $i$ and $p\in \mathrm{AP}(\varphi)$.

\begin{remark}
    Note that $\varepsilon$ must be chosen to be small enough for the completeness of the encoding  (\cref{thm:completeness}).
    Thereafter we assume that, given a piecewise-linear signal $\sigma$ and an STL formula $\varphi$, $\varepsilon$ is small enough to find a $\delta$-stable partition for $\sigma$ and $\varphi$, and we omit $\varepsilon$ from the constraints for simplicity.
\end{remark}

\myparagraph{Consistency Constraints I: Boolean Connectives%
}
We can directly encode conjunction $\bigwedge_{j=1}^m \psi_j$ in STL by recursively applying the shorthand $\land$ in \cref{sec:shorthand}:
$\textstyle
    \langle \bigwedge_{j=1}^m \psi_j \rangle_i = \langle \psi_1 \rangle_i \land \langle \bigwedge_{j=2}^m \psi_j \rangle_i
$
for each $i\in[1,N]$.
It is known that the following alternative encoding avoids auxiliary variables $\langle \bigwedge_{j=k}^m \psi_j \rangle_i$  (where $k$ varies): for each $i\in[1,N]$,
$\langle \bigwedge_{j=1}^m \psi_j \rangle_i \ge 1 - m + \sum_{j=1}^m \langle \psi_j \rangle_i$ and
$\textstyle\langle \bigwedge_{j=1}^m \psi_j \rangle_i \le \langle \psi_j \rangle_i$.
An encoding for disjunction is given similarly:
$\textstyle      \langle \bigvee_{j=1}^m \psi_j \rangle_i \le \sum_{j=1}^m \langle \psi_j \rangle_i$,
$    \langle \bigvee_{j=1}^m \psi_j \rangle_i \ge \langle \psi_j \rangle_i$.

\myparagraph{Consistency Constraints II: Unbounded Temporal Modalities}
For temporal operators with $I=[0,\infty)$, the following encodings are straightforward.
\begin{align}
    \begin{array}{ll}
        \langle \psi_1 \UntilOp{} \psi_2 \rangle_i = \langle \psi_2 \rangle_{i} \lor (\langle \psi_1 \UntilOp{} \psi_2 \rangle_{i+1} \land \langle \psi_1 \rangle_{i}),\quad
        \\
        \langle \psi_1 \ReleaseOp{} \psi_2 \rangle_i = \langle \psi_2 \rangle_{i} \land (\langle \psi_1 \ReleaseOp{} \psi_2 \rangle_{i+1} \lor \langle \psi_1 \rangle_{i})
        \quad & \quad\text{for each $i\in[1,N-1]$,}
        \\
        \langle \psi_1 \UntilOp{} \psi_2 \rangle_N = \langle \psi_2 \rangle_{N}, \quad
        \langle \psi_1 \ReleaseOp{} \psi_2 \rangle_N = \langle \psi_2 \rangle_{N}
              & \quad\text{for  $i=N$.}
    \end{array}
\end{align}
The encodings for $\Diamond,\Box$ are special cases:
\begin{align*}
     & \langle \Diamond \psi \rangle_i = \langle \psi \rangle_{i} \lor \langle \Diamond \psi \rangle_{i+1},                                         \\
     & \langle \Box \psi \rangle_i = \langle \psi \rangle_{i} \land \langle \Box \psi \rangle_{i+1}
     &                                                                                                      & \quad \text{for each $i \in [1,N-1]$} \\
     & \langle \Box \psi \rangle_N = \langle \psi \rangle_N,
    \langle \Diamond \psi \rangle_N = \langle \psi \rangle_N
     &                                                                                                      & \quad \text{for $i = N$}.
\end{align*}

\myparagraph{Consistency Constraints III: Bounded Temporal Modalities}
This is the most technically involved part. The challenge here is that the stability for $\Box_{[a,b]} \psi$ is not guaranteed by the stability for $\psi$ (similarly for $\Diamond_{[a,b]}\psi$). Therefore we need additional MILP constraints for the stability for  $\Box_{[a,b]} \psi$.

Our encoding is inspired by the results from~\cite{prabhakarAutomaticTraceGeneration2018}; ours is  simpler thanks to our theory in \cref{sec:deltaStablePartitioning} where intervals are all closed.

Recall that we use the variables $S^\psi_{i},P^\psi_{i}$ for this purpose.  We focus on $\Box_{[a,b]} \psi$; the encoding of $\Diamond_{[a,b]} \psi$ is similar.  The  constraints on $S^\psi_{i}$ are as follows.
\begin{align*}\small
    \begin{array}{ll}
S^\psi_0  = 0, \qquad
        \langle \psi \rangle_{i} = 0 \;\Rightarrow\; S^\psi_i = 0, \\
        \langle \psi \rangle_{i} = 1 \;\Rightarrow\; S^\psi_i \geq S^\psi_{i-1} + (\gamma_i - \gamma_{i-1})
         & \quad\text{for each $i\in[1,N]$.}
    \end{array}
\end{align*}
It follows that, for any non-negative real number $L\in [0,\gamma_j)$,
we have $S^\psi_j \leq L$ if and only if there exists $k\in [1, j]$ such that $\langle \psi \rangle_k = 0$ and $\gamma_j - \gamma_k \leq L$.

We proceed to the constraints that describe the relationship between $S^\psi_i$ and the semantics of $\Box_{I} \psi$.
Suppose  $\Gamma = ( \gamma_0,\dots,\gamma_N )$ is  $\delta$-stable for a signal $\sigma$ and $\psi$.
Let us write
$\gamma_{N+1} = \infty$ and
$\langle \psi \rangle_{N+1} = \langle \psi \rangle_{N}$
for simplicity.

We consider consistency for the positive and negative cases separately. For the positive one (i.e.\  $\langle \Box_{[a,b]} \psi \rangle_i = 1$), the following observation is used.

%
\begin{proposition}\label{prop:boundedTemporalSufficientConditionPositive}
    Let $\varphi \equiv \Box_I \psi$ be an STL formula in NNF,  and $\Theta$ be a conservative valuation of $\psi$ in $\Gamma = (\gamma_0,\dots,\gamma_N)$ on a signal $\sigma$.
    Given $i\in[1,N]$, suppose
    $(\Gamma_i + I) \cap (\gamma_{j-1}, \gamma_{j}] \neq \emptyset$ implies $\langle \psi \rangle_j = 1$ for each $j \in [i, N + 1]$.
Then $\sigma^t \models \varphi$ holds for any $t \in \Gamma_i$. \qed
\end{proposition}
 \cref{prop:boundedTemporalSufficientConditionPositive} leads to the following MILP constraint:
\begin{align*}\small
    \begin{array}{ll}
        \lnot \langle \varphi \rangle_{i} \lor
        (\gamma_{i} + b  \leq \gamma_{j-1}) \lor
        (\gamma_{i - 1} + a  > \gamma_{j}) \lor
        \langle \psi \rangle_{j}
         & \; \text{for each $i\in[1,N]$, $j \in [i, N+1]$.}
    \end{array}
\end{align*}
The constraint itself does not follow the MILP format; we can nevertheless express it in MILP using an auxiliary Boolean variable $Z_f$. Specifically, an inequality $f(x) \geq 0$ in a disjunctive constraint is constrained by  $Z_f = 1 \Rightarrow f(x) \geq 0$.


For the consistency in the negative case (i.e.\  $\langle \Box_{[a,b]} \psi \rangle_i = 0$), the counterpart of \cref{prop:boundedTemporalSufficientConditionPositive} also involves $S^\psi_j$. See below; it leads to an MILP constraint much like \cref{prop:boundedTemporalSufficientConditionPositive} does.

\begin{proposition}\label{prop:boundedTemporalSufficientConditionNegative}
    Suppose $\varphi$, $\sigma$, $\Gamma$, and $\Theta$ are as in \cref{prop:boundedTemporalSufficientConditionPositive}.
    For any $t \in \Gamma_i$,
    $\sigma^t \not\models \varphi^\delta$ holds if the following conditions are satisfied for each $j \in [i, N]$:
    \begin{equation}
        \begin{cases}
            S^\psi_j \leq b  - a                                & \text{if $\gamma_j \in (\gamma_{i-1} + b, \gamma_i + b)$}, \\
            S^\psi_j \leq \gamma_j - \gamma_i - a               & \text{if $\gamma_i + b \in [\gamma_{j-1}, \gamma_j]$},    \\
            S^\psi_N \leq \max(0, \gamma_N - \gamma_{i} - a) & \text{if $\gamma_i + b > \gamma_N$}.
        \end{cases}
    \end{equation}
\end{proposition}
\begin{proof}
    Let $j_t \in [i, N+1]$ be the unique index such that $t + b \in [\gamma_{j_t-1}, \gamma_{j_t})$.
    When $j_t \leq N$ and $\gamma_{j_t} < \gamma_i + b$, we have $\gamma_{j_t} \in (\gamma_{i-1} + b, \gamma_i + b)$ and by assumption $S^\psi_{j_t} \leq b - a$. There is $k \in [1, j_t]$ such that $\langle \psi \rangle_k = 0$ and $\gamma_k \geq \gamma_{j_t} - b + a > t + a$.
    We obtain $\Gamma_{k} \cap (t + [a, b]) \neq \emptyset$ and then $\sigma^t \not\models \varphi^\delta$ holds.
    The other cases can be checked in a similar manner.
    \qed
\end{proof}

\begin{remark}
For  \cref{prop:boundedTemporalSufficientConditionPositive},
the converse of the statement does not hold.
    This is because $\sigma^t \models \psi$
does not guarantee
$\langle \psi \rangle_i\coloneqq\Theta(\psi, i)=1$
where $t \in \Gamma_i$---we allow $\langle \psi \rangle_i = 0$ when $\sigma^t \models \psi \land \lnot \psi^\delta$. It is similar for \cref{prop:boundedTemporalSufficientConditionNegative}.
However, this does not affect the completeness of the encoding (\cref{thm:completeness}): while the converse of \cref{prop:boundedTemporalSufficientConditionPositive} does not hold for \emph{fixed} $\Gamma$, in our workflow we also search for $\Gamma$, in which case it is easily shown that the MILP constraints derived from \cref{prop:boundedTemporalSufficientConditionPositive} are complete. The same is true for \cref{prop:boundedTemporalSufficientConditionNegative}.
\end{remark}

The remaining cases ($\varphi \equiv \psi_1 \UntilOp{I} \psi_2$ and $\varphi \equiv \psi_1 \ReleaseOp{I} \psi_2$) can be reduced to the cases for $\Box_{I}$ and $\Diamond_{I}$. It is by the rewriting techniques shown in~\cite{donzeEfficientRobustMonitoring2013}:
\begin{align}
    \psi_1 \UntilOp{[a, b]} \psi_2 \quad\sim{} \quad
     & \Diamond_{[a,b]}\psi_2 \land \Box_{[0, a]} (\psi_1 \UntilOp{} \psi_2),  \\
    \psi_1 \ReleaseOp{[a, b]} \psi_2 \quad\sim{} \quad
     & \Box_{[a,b]}\psi_2 \lor \Diamond_{[0, a]} (\psi_1 \ReleaseOp{} \psi_2).
\end{align}
These equivalences hold in both Boolean and robust semantics.

\myparagraph{Correctness of Encoding}
Let $\mathbf{Enc}_{\mathbf{STL}}(\varphi, N, T, \delta)$ denote the polyhedron defined by the above MILP constraints. It is correct in the following sense; see also the goal we announced in the beginning of \cref{sec:encodeFullSTL}. Its proof is by induction on $\varphi$.

\begin{lemma}\label{prop:correctnessEncStl}
    Let $\varphi$ be an STL formula in NNF, $N \in \mathbb{N}$, $T > 0$ and $\delta > 0$.
    Given an assignment  $\boldsymbol{v}\colon \mathbf{Var}(\varphi, N)\to \Real$  that lies in  $\mathbf{Enc}_{\mathbf{STL}}(\varphi, N, T, \delta)$,
    let $\Gamma$, $\varsigma$ be the time sequence and the timed state sequence determined by $\boldsymbol{v}$, and define a valuation $\Theta$ by
    $\Theta(\psi, i) \Defeq  \langle \psi \rangle_i
    $ (cf.\ \cref{def:conservValuation}).
    Then $\Theta$ is a conservative valuation of $\varphi$ in $\Gamma$ on the signal $\varsigma^\mathrm{pwl}$.
\end{lemma}
\begin{proof}
By induction on the structure of $\varphi$.
\qed
\end{proof}

We define
$\mathbf{Enc}(\varphi, \mathcal{M}, N, T, \delta)$ by the intersection of $\mathbf{Enc}_{\mathbf{STL}}(\varphi, N, T, \delta)$, the MILP encoding $\mathbf{Enc}_{\mathbf{model}}(\mathcal{M}, N, T)$ of a system model $\mathcal{M}$, and $\langle \varphi \rangle_1 = 1$.

\begin{theorem}[soundness]\label{thm:soundness}
    Let $\varphi$ be an STL formula in NNF, $\mathcal{M}$ be a  model with a time horizon $T$,  $N \in \mathbb{N}$ and $\delta > 0$.
If an assignment
   $\boldsymbol{v}$
    lies in $\mathbf{Enc}(\varphi, \mathcal{M}, N, T, \delta)$, then
    the induced $\varsigma^\mathrm{pwl} $ has  $\varsigma^\mathrm{pwl} \in \mathcal{L}(\mathcal{M})$ and $\sem{\varsigma^\mathrm{pwl}, \varphi} \geq 0$. \qed
\end{theorem}

\begin{theorem}[completeness]\label{thm:completeness}
    Assume the setting of \cref{thm:soundness}.
    If there is piecewise-linear $\sigma \in \mathcal{L}(\mathcal{M})$ such that $\sem{\sigma, \varphi} \geq \delta$,
    there is an assignment $\boldsymbol{v}$ that lies in
    $\mathbf{Enc}(\varphi, \mathcal{M}, N, T, \delta)$ for some $N \in \mathbb{N}$. \qed
\end{theorem}
\begin{auxproof}
    \begin{proof}
        There is $N$ and $\Gamma = (\gamma_0, \dots, \gamma_N)$ that is $\delta$-stable for $\sigma$ and $\psi$ for all $\psi \in \Sub(\varphi)$.
        We can induce values for the variables $\mathbf{Var}(\varphi, N)$ from $\Gamma$, $\sigma$, and $\varphi$. \qed
    \end{proof}
\end{auxproof}

\section{System Models and Their
  MILP Encoding}\label{sec:systemModels}
We introduce  the MILP encoding $\mathbf{Enc}_{\mathbf{model}}(\mathcal{M}, N, T)$ for some families of models $\mathcal{M}$.
We introduce an exact encoding for \emph{rectangular hybrid automata (RHAs)}, and an approximate one for \emph{HAs with closed-form solutions}.
We also introduce a refinement of the latter---it is more precise and efficient---restricting to \emph{double integrator dynamics}. The last is useful for automotive examples such as \cref{ex:introLead}.

An encoding of RHAs to MILP is not hard, so we defer it to \cref{sec:rectHA}. We focus on the other two families.




\subsection{HAs with Closed-Form Solutions}\label{sec:closedFormSolution}

\begin{wrapfigure}[7]{r}[0pt]{4.5cm}
    \centering

\vspace{-4em}
\begin{tikzpicture}[scale=.4]
        \begin{axis}
            [
                xlabel={$t$},
                ylabel={$f(t)$},
                label style={font=\Large},
                grid,
                xmin=0,
                xmax=4.5,
                ymin=-1.5,
                ymax=4.3
            ]
            \addplot[blue, no marks, solid, domain=0:4.5, samples=50, line width = 1mm] {-(x-2)^2+4};
            \addplot[red, no marks, solid,line width = 1mm] table[col sep=tab]{fig/plot.csv};
        \end{axis}
    \end{tikzpicture}

    \vspace{-1em}
    \caption{MILP encoding of $f(t)$}
    \label{fig:closedForm}
\end{wrapfigure}
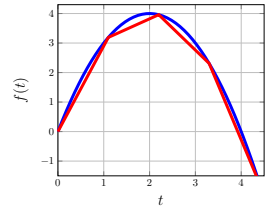
Here we are interested in hybrid automata (HAs) whose continuous flow dynamics at each control mode has a closed-form solution. The basic idea is simple and it is illustrated in \cref{fig:closedForm}, where the solution $f(t)$ of dynamics (blue) is approximated by a piecewise linear function (red). Such MILP encoding
is standard; see e.g.~\cite{asghariTransformationLinearizationTechniques2022}.

We  formalize this intuition.
 Firstly, to accommodate input signals $\tau\in\Signal_{V'}$ (\cref{def:sysModelTraceSet}), we extend the HA definition  so that some variables $x^{\mathrm{in}}$ can be designated to be \emph{input variables}. This means that there are no ODEs whose left-hand side is $\dot{ x^{\mathrm{in}}}$, and that the variable updates associated with mode transitions never change $x^{\mathrm{in}}$.

Then the above ``closed-form solution'' assumption on an HA $\mathcal{H}$ is precisely described as follows.
Let $\vec{x^{\mathrm{in}}}=(x^{\mathrm{in}}_{1},\dotsc,x^{\mathrm{in}}_{k})$ enumerate $\mathcal{H}$'s input variables, and $\vec{x}=(x_{1},\dotsc,x_{l})$ enumerate its other variables.
We assume that, for the flow dynamics at each control mode $u$, there is a \emph{closed-form solution}
\begin{equation}\label{eq:closedForm}
    \vec{x}(t)
    \;=\;
    f_{u}(t,\vec{x^{\mathrm{in}}},\vec{x_{0}})
    \qquad
    \begin{minipage}[t]{.6\textwidth}
        such that, for each $t_{0}\in\Real_{\ge 0}$,  $f_{u}(t_{0},\vec{x^{\mathrm{in}}},\vec{x}_{0})$ is a linear function over the variables $\vec{x^{\mathrm{in}}},\vec{x}_{0}$.
    \end{minipage}
\end{equation}
Here, the variable $t$ is the elapsed time since the arrival at the current control mode $u$; the variables $\vec{x^{\mathrm{in}}}$ refer to the input variables (their values are assumed to be constant within the same mode); and the variables $\vec{x_{0}}$ refer to the \emph{initial} values of $\vec{x}$ on the arrival at $u$. The assumption holds in many examples, such as polynomial dynamics.

Let us motivate the assumption.  A closed-form solution $f_{u}$ helps precision: in piecewise linear  approximation such as in \cref{fig:closedForm}, errors do not accumulate over time; in contrast, if a closed-form solution is not given,  our alternative will be numerical integration e.g.\ by the Euler method, where errors accumulate. The linearity assumption in \cref{eq:closedForm} is there for MILP encoding; see below.

Our approximate MILP encoding poses the closed-form solution assumption and follows the intuition of  \cref{fig:closedForm}. Specifically, 1) it fixes a constant $\varDelta t\in\Real_{\ge 0}$ as a sampling interval; 2) it obtains a family
\begin{math}
    \bigl(\,
    f_{u}(k\cdot\varDelta t, \vec{x^{\mathrm{in}}},\vec{x}_{0})
    \,\bigr)_{k}
\end{math}
of linear functions over the variables $\vec{x^{\mathrm{in}}},\vec{x}_{0}$; and 3) the value of $\vec{x}$ at the elapsed time $t$ is expressed by the linear interpolation
\begin{equation}\label{eq:closedFormApprox}
    \textstyle
    \frac{(k+1)\varDelta t -t}{\varDelta t}
    f_{u}(k\varDelta t, \vec{x^{\mathrm{in}}},\vec{x}_{0})
    +
    \frac{t-k\varDelta t}{\varDelta t}
    f_{u}\bigl((k+1)\varDelta t, \vec{x^{\mathrm{in}}},\vec{x}_{0}\bigr),
\end{equation}
where $k$ is such that $k\varDelta t\le t \le (k+1)\varDelta t$. This encoding of flow dynamics is combined with the HA structure, much like in \cref{sec:rectHA}, yielding an approximate MILP encoding of the whole HA\@.

The above encoding has two sources of numerical errors.
One  is  \emph{linear interpolation}. Errors caused by it are  illustrated  in \cref{fig:closedForm} as the vertical margin between blue and red.

The other source is  \emph{binary expansion}~\cite{gloverImprovedLinearInteger1975,gupteSolvingMixedInteger2013}, a standard MILP technique for encoding bilinear functions.
Indeed, in~\cref{eq:closedFormApprox},  $t, \vec{x^{\mathrm{in}}},\vec{x}_{0}$ are all continuous variables in MILP, and the expression~\cref{eq:closedFormApprox} can contain their products. The linearity assumption in~\cref{eq:closedForm} has been posed to restrict~\cref{eq:closedFormApprox} to bilinear.





\subsection{HAs with Double Integrator Dynamics}\label{sec:piecewiseConstAccel}
Our next focus is a special case of the model family of \cref{sec:closedFormSolution}, where each continuous flow is \emph{double integrator dynamics}. This is important because 1) it gets rid of one of the two error sources in \cref{sec:closedFormSolution}, namely  linear interpolation, by the \emph{trapezoidal rule}, and 2) it can be used for many automotive dynamics models (cf.\ \cref{ex:introLead}).

The \emph{trapezoidal rule} is a basic technique in numerical integration~\cite{Atkinson89}, where $\int_{a}^{b}g(t)\,\mathrm{d}t$ is approximated by $(b-a)\frac{g(a)+g(b)}{2}$. For double integrator dynamics, we apply the trapezoidal rule to the velocity $v$, and it is exact since $v$'s evolution is linear. This allows us to express the position $x$ in the bilinear form
\begin{math}
    x= t\cdot\frac{v_{0} + v}{2}
\end{math}, using the variables $t$ (elapsed time), $v_{0}$ (initial velocity), and $v$ (current velocity). Thus we can  dispose of  the sampling points and their interpolation~\cref{eq:closedFormApprox} in \cref{sec:closedFormSolution}.

We note that this specific modeling is still approximate since the second error source in \cref{sec:closedFormSolution}, namely binary expansion, remains. Nevertheless, it is more precise and efficient (piecewise linear approximation in \cref{sec:closedFormSolution} is costly, too). We exploit this encoding for our automotive case studies such as \cref{ex:introLead}.

\section{Implementation and Experiments}\label{sec:expr}
We implemented, in Python, our MILP encodings of the STL semantics  (\cref{sec:encodeSTL}) and two model families, namely RHAs (\cref{sec:rectHA}) and double integrator dynamics (\cref{sec:piecewiseConstAccel}; multiple modes are not supported since  our benchmarks do not need them). The hyperparameter $\delta$ in our encoding is fixed at $0.1$ for all benchmarks. The resulting MILP constraints are solved by Gurobi Optimizer~\cite{gurobi}. This prototype implementation is called \emph{STLts}---STL trace synthesizer.

Our experiments are designed to address the following research questions.
\begin{description}
    \item[RQ1] Assess the effect of  variability bounds $N$ (\cref{prob:traceSynthesis}) on the performance.
    \item[RQ2] Compare the performance of STLts with optimization-based falsification, and with SMT-based model checking.
    \item[RQ3] Assess the performance  of STLts for   real-world complex scenarios.
    \item[RQ4] Assess the performance of STLts in parameter mining (\cref{prob:existParMining}).
\end{description}

We used three classes of benchmarks: \emph{rear-end near collision (RNC)}, \emph{navigation (NAV)}, and \emph{disturbance scenarios in ISO 34502 (ISO)}. In each class, we have multiple STL specs, resulting in benchmarks such as RNC1, RNC2, etc.

\myparagraph{Rear-End Near Collision (RNC1--3)}
As discussed in \cref{ex:introLead}, these automotive benchmarks are simplifications of the ISO benchmarks  below. The spec $\mathtt{RNC1}$ is presented in \cref{ex:introLead}. The system model~\cref{eq:introLeadModel} (see also
 $\MRNC$ in \cref{ex:introLeadModel}) is double integrator dynamics (\cref{sec:piecewiseConstAccel}) and is shared by the benchmarks RNC1--3.

The other two specs $\mathtt{RNC2}, \mathtt{RNC3}$ are defined as follows, using formulas in~\cref{eq:RNC1}:
\begin{align}\label{eq:RNC23}
    \begin{array}{rll}
        \mathtt{RNC2}      & \quad:\equiv\quad & \bigl(\Box (x_{\front} - x_{\rear} \ge 0)\bigr) \land
        \\
                           &                   & \quad
        \Diamond_{[0, 9]} \bigl(
        (\Box_{[0,1]} \mathtt{danger} ) \land
        (\Box_{[0,1]} a_{\rear} \ge 1 ) \land
        (\Diamond_{[1,5]} \lnot\mathtt{danger})
        \bigr)                                                                                                                                                             \\
        \mathtt{trimming2} & \quad:\equiv\quad & (\Diamond \mathtt{danger}) \Rightarrow \bigl((\Box_{[0, 1]} a_{\rear} \geq 1) \mathbin{\mathcal{U}} \mathtt{danger}\bigr)
        \\
        \mathtt{RNC3}      & \quad:\equiv\quad & \Box(\mathtt{dyn\_inv}\land\mathtt{trimming2} ) \land \Diamond_{[0, 9]}\Box_{[0,1]}\mathtt{danger}
    \end{array}
\end{align}

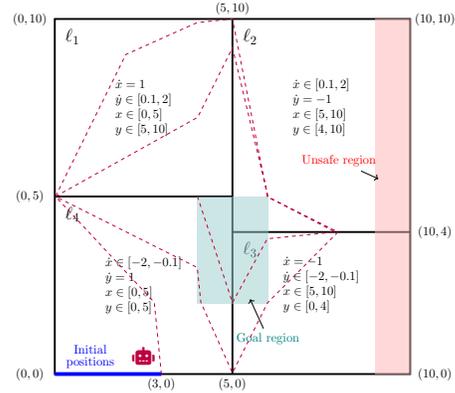
\begin{wrapfigure}[15]{r}[0pt]{5.5cm}
    \centering

    \vspace{-1em}
    \centering
    \resizebox{0.5\textwidth}{!}{
        \begin{tikzpicture}
            \draw[line width=1.5pt] (0,10) rectangle (10,0);
            \draw[line width=1.5 pt] (0,5) -- (5,5);
            \draw[line width=1.5 pt] (5,4) -- (10,4);
            \draw[line width=1.5 pt] (5,0) -- (5,10);
            \node at (-0.7, 10) {$(0,10)$};
            \node at (-0.7, 0) {$(0,0)$};
            \node at (10.7, 0) {$(10,0)$};
            \node at (10.7, 10) {$(10,10)$};
            \node at (-0.7,5) {$(0,5)$};
            \node at (10.7,4) {$(10,4)$};
            \node at (5,-0.3) {$(5,0)$};
            \node at (5,10.3) {$(5,10)$};
            \node at (3,-0.3) {$(3,0)$};

            \draw[line width=3pt, color = blue] (0,0) -- (3,0); %
            \node[align=center,color=blue] at (1,0.5) {Initial \\[-0.1cm] positions};
            \node[color=purple] at (2.5 ,0.5) {\Large \faRobot};

            \node at (0.5, 9.5) {\Large $\ell_1$};
            \node[align=left, font=\normalsize] at (2.5, 7.5) {$\dot{x}=1$ \\ $\dot{y} \in [0.1,2]$ \\ $x\in[0,5]$ \\ $y \in [5,10]$};

            \node at (5.5, 9.5) {\Large $\ell_2$};
            \node[align=left, font=\normalsize] at (7.5, 7.5) {$\dot{x} \in [0.1, 2]$ \\ $\dot{y} = -1$ \\ $x\in[5,10]$ \\ $y \in [4,10]$};

            \node at (5.5, 3.5) {\Large $\ell_3$};
            \node[align=left, font=\normalsize] at (7.5, 2.5) {$\dot{x} = -1$ \\ $\dot{y} \in [-2, -0.1]$ \\ $x\in[5,10]$ \\ $y \in [0,4]$};

            \node at (0.5, 4.5) {\Large $\ell_4$};
            \node[align=left, font=\normalsize] at (2.5, 2.5) {$\dot{x} \in [-2,-0.1]$ \\ $\dot{y}=1$ \\ $x\in[0,5]$ \\ $y \in [0,5]$};

            \fill[fill=red!30, fill opacity=0.5] (9,10) rectangle (10,0); 
            \node[align=center, color=red] (A) at (8,6) {Unsafe region};
            \draw[->] (A) -- (9.1,5.5);
            \fill[fill=teal!40, fill opacity=0.5] (4,5) rectangle (6,2); 
            \node[align=center, color=teal] (B) at (6,1) {Goal region};
            \draw[->] (B) -- (5.5, 2.1);

            \draw[dashed,color=purple] (3.0,0) -- (2.80,2.02) -- (0.02,5.00) -- (2.00,9.00) -- (4.01,9.90) -- (5.00,10.00) -- (5.99,4.99) -- (7.97,4.00) -- (5.99,2.01) -- (5.00,0.03) -- (4.11,2.02) -- (4.01,3.00) -- (0.02,5.00) -- (4.02,7.23) -- (5.00,9.19) -- (5.98,5.09) -- (5.99,4.99) -- (7.89,4.0) -- (5.98,3.81) -- (5.0,2.01) -- (4.02,4.99);
        \end{tikzpicture}
    }

    \vspace{-1em}
    \caption{The RHA $\MNAV$ for NAV1--2}
    \label{fig:example_nav}
\end{wrapfigure}
\myparagraph{Navigation (NAV1--2)}
Here we use a system model that adapts NAV-2 from \cite{duggiralaAbstractionRefinementStability2011}. The latter is  a standard example of an RHA, used e.g.\ in~\cite{buARCHCOMP22CategoryReport}.

Our system model $\MNAV$ is an RHA that describes the motion of a point robot in a $2 \times 2$ grid where each region has a rectangular vector field, with a time horizon $T = 40$. See \cref{fig:example_nav}.
We have $4$ regions $\ell_1,\dotsc,\ell_4$, each associated with rectangular bounds for $\dot{x},\dot{y}$ and invariants; besides, we set an unsafe region $\mathtt{unsafeR}$ ($x\in [9,10]$) and a goal region $\mathtt{goalR}$ ($x\in[4,6] \wedge y\in[2,5]$). The robot starts from an initial position $(x_0,y_0)$ where $x_0\in[0,3]\wedge y_0=0$.


\begin{auxproof}
    The model describes the motion of a point robot in an $n \times n$ grid where each region is associated with a rectangular vector field. The robot can move in and between the fields according to the flow equations of them. It has been a standard benchmark for rectangular hybrid automata~\cite{buARCHCOMP22CategoryReport}. As shown in~\cref{fig:example_nav}, we modified the NAV-2 from \cite{duggiralaAbstractionRefinementStability2011}. The grid has been divided into $4$ regions $\ell_1,\ell_2,\ell_3$ and $\ell_4$. The motion of the robot is according to the flow equations in each region. For instance, in $\ell_1$, the dynamics $\dot{x}=1$ and $\dot{y}\in[0.1,2]$ restrict the robot behaviors in horizontal ($x$) and vertical ($y$) directions respectively, and $x\in[0,5]$ and $y\in[5,10]$ are the invariants. In the grid, there are an unsafe region $\mathtt{unsafeR}$ ($x\in [9,10]$) and a goal region $\mathtt{goalR}$ ($x\in[4,6] \wedge y\in[2,5]$). The robot starts from an initial position $(x_0,y_0)$ where $x_0\in[0,3]\wedge y_0=0$. We consider two specifications presented with the corresponding STL formulas in \cref{tab:NAV_specification}.
\end{auxproof}

We consider two specs: $\mathtt{NAV1}\,:\equiv\,\Diamond (\Box_{[0,3]} ((x, y) \in \mathtt{goalR})) \land \Box( x \not\in \mathtt{unsafeR})$ and $\mathtt{NAV2}\,\equiv\,\Box((x,y) \in \ell_3 \to \Diamond_{[0,3]} (x,y) \in \ell_4)$.
$\mathtt{NAV1}$ is almost a standard reach-avoid constraint, but it additionally requires the \emph{persistence} to the goal region for three seconds. Such specifications are  not accommodated in many control and model checking frameworks specialized in reach-avoid constraints (see e.g.~\cite{buARCHCOMP22CategoryReport}).
$\mathtt{NAV2}$ is a \emph{response specification}---the trigger (being in $\ell_3$) must be responded by moving to $\ell_4$ within a three-second deadline. Such specs are common in manufacturing; see e.g.~\cite{ZhangHA19}.

\clearpage

\begin{wraptable}[11]{r}[0pt]{6cm}
    \vspace{-1em}
    \caption{Disturbance scenarios in the ISO 34502 standard.  Table from~\cite{ISO34502}}
    \label{table:iso}
    \centering
    \includegraphics[width=.5\textwidth]{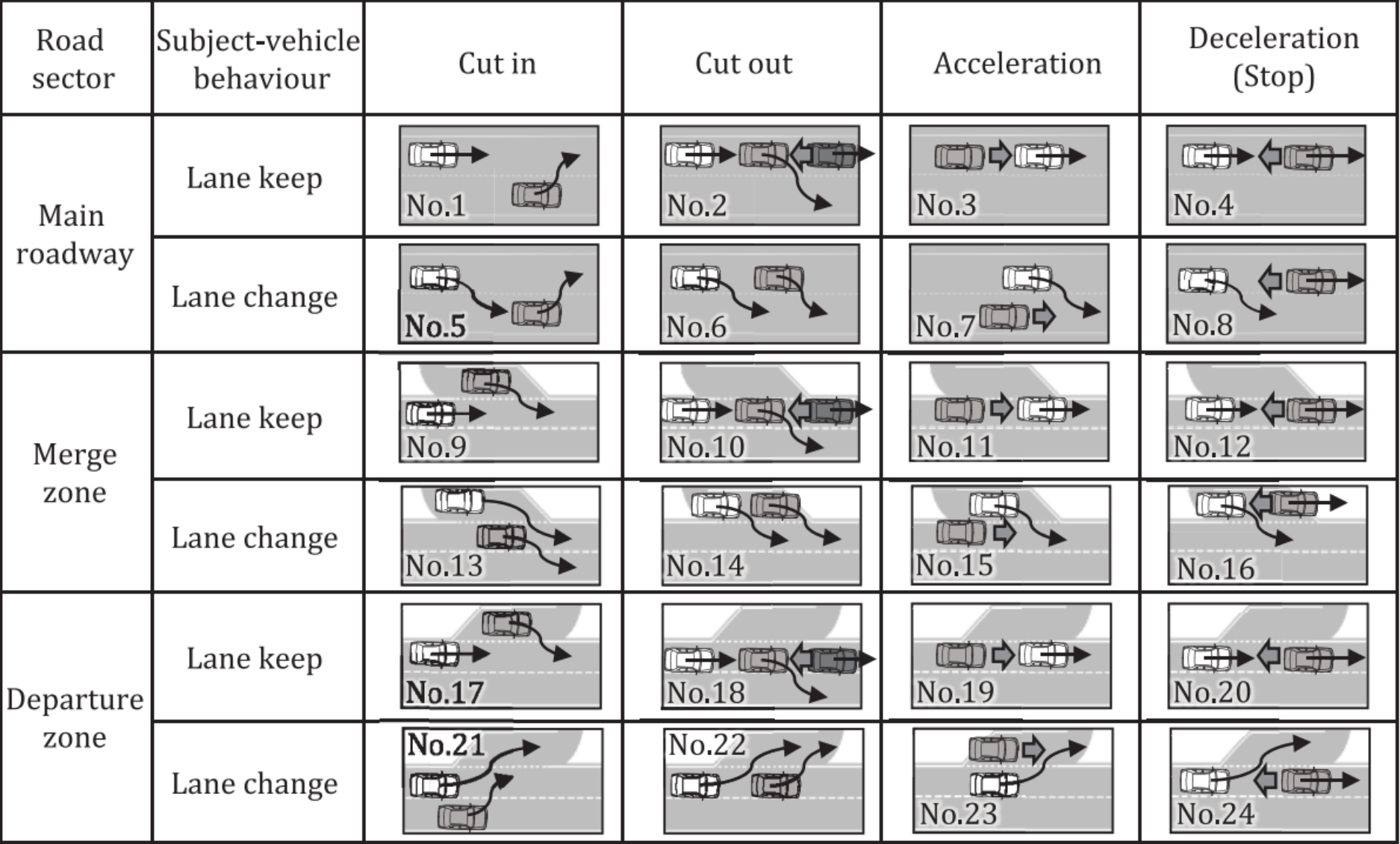}

\end{wraptable}
\myparagraph{ISO 34502 Disturbance Scenarios for Automated Driving (ISO1, ISO3, $\dotsc$, ISO8)}
These benchmarks motivated the current work. As discussed in \cref{sec:intro} (see \cref{ex:introLead}), we obtained in~\cite{ReimannMHBCSWAHUY24_toAppear} complex STL specs as the formalization of the disturbance scenarios in the ISO 34502 standard, but in our illustration efforts by trace synthesis, we found that existing techniques such as optimization-based falsification struggle.

In our experiments, the system model is similar to $\MRNC$ (\cref{ex:introLead,ex:introLeadModel}), while lateral dynamics is added and the time horizon is $10$ time units here.
As for specs, we use seven STL specs $\mathtt{ISO1}, \mathtt{ISO3},\dotsc,\mathtt{ISO8}$; these are obtained in~\cite{ReimannMHBCSWAHUY24_toAppear} as the formalization of the \emph{disturbance scenarios No.\ 1,3,$\dotsc$,8}  in the ISO 34502 standard for automated driving vehicles. See \cref{table:iso}.  Scenario No.\ 2 was omitted in~\cite{ReimannMHBCSWAHUY24_toAppear} since it involves three vehicles; we omit Scenarios No.\ 9--24 since they are the same with No.\ 1--8 except in the road shape.

Specifically, the specs $\mathtt{ISO}i$ follow the common format shown below~\cite{ReimannMHBCSWAHUY24_toAppear}:
\begin{displaymath}\small
    \begin{array}{rcl}
        \mathtt{ISO}{i}
         & \;\equiv\; & \mathtt{initSafe}
        \wedge \mathtt{disturb}_i,
        \qquad                            \\
        \mathtt{disturb}_i
         & \;\equiv\; &
        \mathtt{initialCondition}_i
        \wedge \mathtt{behaviourSV}_i
        \wedge \mathtt{behaviourPOV}_i
    \end{array}
\end{displaymath}
where SV refers to the subject (``ego'') vehicle and POV refers to the principal other vehicle. The component formulas $ \mathtt{initialCondition}_i$, $\mathtt{behaviourSV}_i$ and $ \mathtt{behaviourPOV}_i$ vary for different scenarios (No.\ $i$). Going into their definitions are beyond the scope of this paper; we highlight $\mathtt{ISO5}$ as an example to demonstrate the complexity of the specs $\mathtt{ISO}{i}$.
\begin{equation}\label{eq:ISOcomplex}\small
    \begin{array}{rcl}
        \mathtt{initialCondition}_5
         & \;\equiv\; & \top
        \qquad\qquad\qquad
        \mathtt{behaviourSV}_5
        \;\equiv\;  \mathtt{leavingLane}(\mathtt{SV},L)
        \\
        \mathtt{behaviourPOV}_5
         & \;\equiv\; & \mathtt{cutIn}(\mathtt{POV},\mathtt{SV})
        \\
        \mathtt{leavingLane}(a,L)
         & \;\equiv\; &
        \mathtt{atLane}(a,L) \wedge \Diamond(\neg\mathtt{atLane}(a,L))
        \\
        \mathtt{cutIn}(\mathtt{POV},\mathtt{SV},L)
         & \;\equiv\; &
        \neg\mathtt{sameLane}(\mathtt{POV},\mathtt{SV},L)
        \wedge \Diamond\bigl(\mathtt{danger}(\mathtt{SV},\mathtt{POV})
        \\&&
        \!\!\!\!\!\!   \wedge \Diamond_{[0,\mathtt{minDanger}]}(\mathtt{sameLane}(\mathtt{SV},\mathtt{POV},L)  \wedge \mathtt{aheadOf}(\mathtt{SV},\mathtt{POV}))\bigr)
        \\
        \mathtt{danger}(\mathtt{SV},\mathtt{POV})
         & \;\equiv\; &
        \Box_{[0,\mathtt{minDanger}]} \mathtt{rssViolation}(\mathtt{SV},\mathtt{POV})
    \end{array}
\end{equation}
The formulas not defined here are suitably defined atomic propositions.

\myparagraph{Experiment Settings}
Our implementation \emph{STLts} is compared with the following tools: 1) a widely used optimization-based falsification tool \emph{Breach}~\cite{Donze10}; 2) another falsification tool \emph{ForeSee}~\cite{foreSeeByZhenya,ZhangLAMHZ21}  that emphasizes optimized treatment of Boolean connectives in STL; 3) an MILP-based STL optimal control tool \emph{bluSTL}~\cite{donzeBluSTLControllerSynthesis}; and 4) \emph{STLmc}, an SMT-based bounded STL model checker~\cite{YuLB22}.

The experiments were conducted on an Amazon EC2 c4.4xlarge instance (2.9 GHz Intel Xeon E502666 v3, 30.0GB RAM) running Ubuntu Server 20.04.

\clearpage

\myparagraph{RQ1: the Effect of the Variability Bound $N$}

\begin{wrapfigure}[7]{r}[0pt]{4.5cm}
    \centering

    \vspace{-4em}
    \includegraphics[width=4cm]{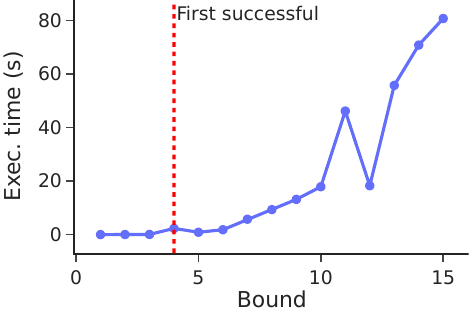}

    \vspace{-1em}
    \caption{Execution time of STLts for different var.\ bd.\ $N$, on $\mathtt{ISO6}$}
    \label{fig:runtimeVaryingN}
\end{wrapfigure}
\noindent There is an obvious trade-off about the choice of a variability bound $N$ (\cref{prob:traceSynthesis}): bigger $N$ means the search is more extensive, but it incurs greater computational cost.

This tendency is confirmed in our experiments;  the result for the $\mathtt{ISO6}$ benchmark is in \cref{fig:runtimeVaryingN} for illustration. Here,  synthesis was successful for $N=4$ for the first time.

We also observe in the figure that computational cost is low when trace synthesis is unsuccessful. This suggests the following strategy: we start with small $N$ and increment it if trace synthesis is unsuccessful. We might waste time by trying too small $N$'s; but the wasted time should be small.

\begin{wraptable}[15]{r}[0pt]{5.5cm}

    \vspace{-3em}
    \caption{Experimental results for trace synthesis, showing execution time (seconds). $(N)$ for STLts is the first successful bound. Timeout (t/o) is 600 sec. }
    \label{table:experimentalResults}
    \scalebox{.7}{\begin{tabular}{lrrrrrr}\toprule
                                        & \multicolumn{2}{c}{STLts}  & Breach
                                        & ForeSee
                                        & bluSTL
                                        & STLmc
            \\\midrule
            $\mathtt{RNC1}$\phantom{ho} & \cellcolor{LightCyan} 0.1 & (3)    & 59.4                   & 546.8                          & $(\P)$                        & t/o
            \\
            $\mathtt{RNC2}$             & \cellcolor{LightCyan}0.3  & (4)    & 9.3                   & 104.3                       & 14.3                          & t/o
            \\
            $\mathtt{RNC3}$             & \cellcolor{LightCyan}0.1  & (3)    & 81.3                   & 197.4                          & $(\P)$                        & t/o
            \\\midrule
            $\mathtt{NAV1}$             & 32.5                    & (17)   & \multirow{2}{*}{$(*)$} & \multirow{2}{*}{$(*)$}       & \multirow{2}{*}{$(\ddagger)$} & \cellcolor{LightCyan}16.5
            \\
            $\mathtt{NAV2}$             & \cellcolor{LightCyan}2.1   & (11)   &                        &                              &                               & 10.0
            \\\midrule
            $\mathtt{ISO1}$             & \cellcolor{LightCyan}0.4 & (3)    & 8.9                    & t/o & \multirow{7}{*}{$(\dagger)$}  & \multirow{7}{*}{$(\dagger)$}
            \\
            $\mathtt{ISO3}$             & \cellcolor{LightCyan}0.2  & (2)    & t/o                    &   t/o                           &                               &
            \\
            $\mathtt{ISO4}$             & \cellcolor{LightCyan}0.4  & (2)    & t/o                    &        t/o                      &                               &
            \\
            $\mathtt{ISO5}$             & \cellcolor{LightCyan}9.9  & (4)    & 31.2                   &           435.8                   &                               &
            \\
            $\mathtt{ISO6}$             & \cellcolor{LightCyan}2.4  & (4)    & t/o                    &                58.9              &                               &
            \\
            $\mathtt{ISO7}$             & \cellcolor{LightCyan}0.6  & (3)    & 33.6                   &              187.2               &                               &
            \\
            $\mathtt{ISO8}$             & \cellcolor{LightCyan}1.5  & (3)    & 38.8                   &                t/o              &                               &
            \\\bottomrule
        \end{tabular}}
\end{wraptable}
\myparagraph{Experimental Results, Overview}
Our experimental results are in summarized in \cref{table:experimentalResults}, where the best performers are highlighted by color.

We explain the missing entries. In $(*)$, the tool is not applicable due to the nondeterminism of the benchmark. In $(\dagger)$, we did not conduct experiments since the performance comparison with STLts is already clear with simpler $\mathtt{RNC}$ benchmarks. In $(\ddagger)$,  bluSTL does not support multiple control modes. $(\P)$ is because bluSTL (at least its implementation available to us) does not support the until $\mathcal{U}$ modality.

Overall, our STLts is clearly the best performer in all  benchmarks but one. The other tools time out, or takes tens of seconds. For our motivation of \emph{illustrating} STL specs by trace synthesis in close interaction with users, tens of seconds is prohibitively long. The results adequately demonstrate satisfactory performance of our algorithm, in trace synthesis for complex STL specs.

\begin{table}[tbp]
\newcommand{\pie}[1]{%
    \begin{tikzpicture}[scale=.8, baseline=-3pt]
        \draw (0,0) circle (1ex);
        \fill (0,0) -- (90:1ex) arc (90:90-#1:1ex) -- cycle;
    \end{tikzpicture}%
}

\newcommand{\feature}[2][]{%
    \ifthenelse{\equal{#1}{nohl}}{}{%
        \ifthenelse{\equal{#2}{++}}{\hl}{}%
        \ifthenelse{\equal{#2}{+}}{\hl}{}%
    }%
    \ifthenelse{\equal{#2}{++}}{\pie{360}}{}%
    \ifthenelse{\equal{#2}{+}}{\pie{180}}{}%
    \ifthenelse{\equal{#2}{-}}{\pie{90}}{}%
    \ifthenelse{\equal{#2}{--}}{-}{}%
}

\caption{Comparison of our approach (STLts) with baselines (Breach, ForeSee, bluSTL, STLmc).
    Highlited cells represent positive features.}\label{table:comparison}
\resizebox{\textwidth}{!}{\begin{threeparttable}
\begin{tabular}{>{\raggedright}p{10em}|| >{\raggedright}p{9em}| >{\raggedright}p{9em}| >{\raggedright}p{9em} | >{\raggedright\arraybackslash}p{9em}}

Feature &
\bf STLts &
\bf Breach/ForeSee &
\bf bluSTL &
\bf STLmc \\
\hline

\bf Trace synthesis for analyzing specs  &
\feature{++}\; Successful in all benchmarks with large STL formulas &
\feature{+}\; Good for falsifying models but not good with large STL formulas &
\feature{--}\; Timeout in most of benchmarks &
\feature{-}\; Timeout except for linear dynamics \\
\hline

\bf Model checking  &
\feature{+}\; Complete up to $N$ and $\delta$  &
\feature{--}\;  &
\feature{+}\; Control synthesis with guarantee &
\feature{++}\; Complete up to $N$ \\
 \hline

\bf Parameter mining &
\feature{+}\; By MILP &
\feature{--}\;  &
\feature{+}\; By MILP &
\feature{-}\; By binary search \\
 \hline

\bf Continuous STL semantics &
\feature{++}\; Variable-interval encoding &
\feature{--}\; Discretized &
\feature{--}\; Discretized &
\feature{++}\; Variable-interval encoding \\
 \hline

\bf Accommodated class of nonlinear dynamics &
 MILP-encodable, can be nondeterministic &
Black-box, deterministic &
 MILP-encodable, can be nondeterministic &
 SMT-encodable, can be nondeterministic \\
 \hline

 \end{tabular}
 \begin{tablenotes}
    \item
    \feature{++} = full support; \feature{+} = partial support; \feature{-} = very limited support; \feature{--} = not supported.
    \end{tablenotes}    \end{threeparttable}
 } 
\end{table}

\myparagraph{RQ2:  Comparison with Other Approaches} A summary of comparison is in \cref{table:comparison}.
The comparison with optimization-based falsification tools is as we expected---their struggle with complex specs motivated this work (\cref{sec:intro}). Boolean connectives in STL specs have been found problematic in falsification:
this is called the \emph{scale problem}~\cite{ZhangHA19,ZhangLAMHZ21}. The results in \cref{table:experimentalResults} show that our benchmark specs are even beyond the capability of ForeSee, a tool that incorporates Monte Carlo tree search to specifically handle the scale problem. After all, one can say that falsification tools are aimed at complex \emph{models}, while our STLts is aimed at complex \emph{specs}.

STLmc has a similar (``dual'') scope and utilizes a similar technique (stable partitioning) to our STLts; the main difference is that STLmc is SMT-based while STLts is MILP-based. Therefore STLts accommodates a smaller class of models, but it can be faster on them exploiting numeric optimization. \cref{table:experimentalResults} suggests the advantage of STLts for common STL specs in manufacturing.

\myparagraph{RQ3: Performance in Real-World Scenarios} For this RQ, we refer to STLts's performance on the ISO benchmarks. Illustrating the specs $\mathtt{ISO}i$ by trace synthesis is a real-world problem about safety standards for automated driving (\cref{sec:intro}), and \cref{table:experimentalResults} shows that STLts has sufficient performance and scalability to handle complex specs there (see~\cref{eq:ISOcomplex}).

\myparagraph{RQ4: Performance in Parameter Mining}

\begin{wrapfigure}[10]{r}[0pt]{4.5cm}

    \vspace{-4em}

%
%
%
%
%

    \includegraphics[width=4.5cm]{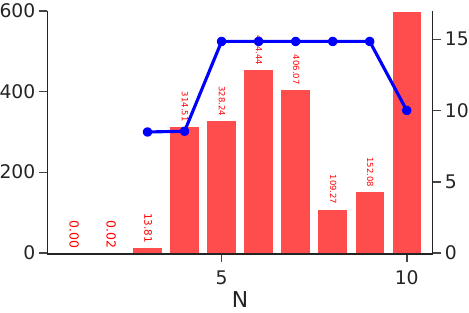}
    \vspace{-2.5em}

    \caption{STLts for parameter synthesis. Red is execution time (axis left, seconds); blue is the maximum $p$ (axis right). }
    \label{fig:paramSynthesis}
\end{wrapfigure}
\noindent
We conducted parameter mining experiments with the $\mathtt{ISO8}$ benchmark. Its specification
has a subformula $\mathtt{fasterThan}(SV,POV,p)$ that requires that SV's velocity is bigger than POV's by at least a parameter $p$. We used STLts to solve \cref{prob:existParMining}, that is, to find the maximum $p$ for which a satisfying trace exists.

\cref{fig:paramSynthesis} shows the results with varying variability bound $N$.
Parameter mining is generally more expensive than trace synthesis. This is because the former has a nontrivial objective function (namely $p$ in this example), while the latter does not (it is thus a constraint satisfaction problem).
We observe the optimization with $N \geq 10$ resulted in a timeout.
The tendency, much like in trace synthesis, is that the result (max $p$) improves but execution time gets larger as $N$ becomes bigger (there are some exceptions such as $N=8,9$ though).
Taking the same strategy as above (incrementing $N$), it takes roughly 10 minutes to obtain a largely converged value ($\sim 14.9$ for the maximum $p$). Overall, we believe this is a realistic performance for practical usage.


%
%

%
%
%

\bibliographystyle{splncs04}
\bibliography{export}

\appendix

\section{Further Details}
\label{sec:appendix}

\begin{definition}[STL (Boolean) semantics]\label{def:STLBoolSem}
    Let $\signalSymbol$ be a signal and  $\varphi$ be an STL formula, both over $V$. The \emph{satisfaction relation} $\signalSymbol\models\varphi$ between them is defined as follows; the semantics of the other operators are defined similarly~\cite{fainekosRobustnessTemporalLogic2009}.
    \begin{align*}
        \begin{array}{l}
            \BoolSat{\signalSymbol}{p} \;\Defiff \;
            \pi_p\bigl(\signalSymbol(0)\bigr)\ge 0
            \qquad
            (\BoolSat{\signalSymbol}{\bot} \text{ never holds})
            \qquad
            \\
            \BoolSat{\signalSymbol}{\neg \varphi} \;\Defiff\;  \signalSymbol\not\models\varphi                                                                \\
            \BoolSat{\signalSymbol}{\varphi_1 \wedge \varphi_2} \;\Defiff\; \BoolSat{\signalSymbol}{\varphi_1}\text{ and } \BoolSat{\signalSymbol}{\varphi_2} \\
            \BoolSat{\signalSymbol}{\varphi_1 \UntilOp{I} \varphi_2} \;\Defiff\;
            \exists t \in I.\,
            (\BoolSat{\signalSymbol^t}{\varphi_2} \;\land\; \forall t'\in [0,t).\,
            \BoolSat{\signalSymbol^{t'}}{\varphi_1})
            \\
            \BoolSat{\signalSymbol}{\varphi_1 \ReleaseOp{I} \varphi_2} \;\Defiff\;
            \forall t \in I.\,
            (\signalSymbol^t\not\models\varphi_2 \;\Rightarrow\; \exists t'\in [0,t).\,
            \BoolSat{\signalSymbol^{t'}}{\varphi_1})
        \end{array}
    \end{align*}
\end{definition}

\begin{definition}[STL robust semantics]\label{def:STLRobustSem}
    Let $\signalSymbol$ be a signal and  $\varphi$ be an STL formula, both over $V$. STL \emph{robust semantics} returns a quantity $\Robust{\signalSymbol}{\varphi} \in \R \cup\{\infty, -\infty\}$ that indicates the satisfaction level of $\signalSymbol$ to $\varphi$, defined as follows.
    \begin{align*}
        \begin{array}{l}
            \Robust{\signalSymbol}{p} \;\Defeq \;
            \pi_p\bigl(\signalSymbol(0)\bigr)
            \qquad
            \Robust{\signalSymbol}{\bot} \;\Defeq\; -\infty
            \qquad
            \Robust{\signalSymbol}{\neg \varphi} \;\Defeq\; - \Robust{\signalSymbol}{\varphi}                                                                \\
            \Robust{\signalSymbol}{\varphi_1 \wedge \varphi_2} \;\Defeq\; \min\big(\Robust{\signalSymbol}{\varphi_1}, \Robust{\signalSymbol}{\varphi_2}\big) \\
            \Robust{\signalSymbol}{\varphi_1 \UntilOp{I} \varphi_2} \;\Defeq\;
            {\sup_{t \in I}\left(\min\left(
            \,\Robust{\signalSymbol^t}{\varphi_2}, 
            \inf_{t'\in[0, t)} \Robust{\signalSymbol^{t'}}{\varphi_1}\,
            \right)
            \right)}                                                                                                                                         \\
            \Robust{\signalSymbol}{\varphi_1 \ReleaseOp{I} \varphi_2} \;\Defeq\;
            {\inf_{t \in I}\left(\max\left(
            \,\Robust{\signalSymbol^t}{\varphi_2}, 
            \sup_{t'\in[0, t)} \Robust{\signalSymbol^{t'}}{\varphi_1}\,
            \right)
            \right)}
        \end{array}
    \end{align*}
    The semantics of the other operators are defined similarly~\cite{fainekosRobustnessTemporalLogic2009}.
\end{definition}

It is well-known that,    by the quantitative robust semantics, one can infer the Boolean semantics:
if $\Robust{\signalSymbol}{\varphi}$ is positive, it implies that  $\signalSymbol\models \varphi$, and if $\Robust{\signalSymbol}{\varphi}$ is negative, it implies that $\signalSymbol\not\models \varphi$.

\begin{definition}[PSTL]\label{def:PSTL}
    Let $\vec{p}=(p_{1},\dotsc,p_{g})$ and $\vec{q}=(q_{1},\dotsc,q_{h})$ be vectors of syntactic parameters; those in $\vec{p}$ are called \emph{magnitude parameters} and those in $\vec{q}$ are \emph{timing parameters}.

    The syntax of \emph{parametric STL (PSTL)} is obtained by extending that of STL (\cref{def:stlSyntax}) as follows: 1) atomic propositions can also be in the form $\alpha :\equiv (f(\vec{w}) \geq  p_{i})$, having a magnitude parameter $p_{i}$ on the right-hand side instead of $0$; and 2) allowing a timing parameter $q_{j}$ as a bound of the interval $I=[a,b]$ that indexes a temporal operator $\UntilOp{I}$ or $\ReleaseOp{I}$  (i.e.\ $a$ or $b$ can be $q_{j}$, instead of a constant).

    Let $P\subseteq \Real^{g}$ and $Q\subseteq \Real^{h}$; these are the \emph{value domains} of the parameters $\vec{p},\vec{q}$.
    Let $\vec{u}=(u_{1},\dotsc,u_{g})\in P$ and $\vec{v}=(v_{1},\dotsc,v_{h})\in Q$ be vectors of real numbers from the domains. Given a PSTL formula $\varphi$, by replacing the occurrences of $p_{i}$ and $q_{j}$ with $u_{i}$ and $v_{j}$,  we obtain  an STL formula. It is denoted by $\varphi_{\vec{u},\vec{v}}$.
\end{definition}

The following is an easy consequence of \cref{def:finiteVariability}: $\mathcal{P}$ is obtained as a common refinement of partitions for subformulas.

    \begin{proposition}\label{prop:stablePartitioningExists}
        If  $\sigma$ is finitely variable with respect to $\varphi$, then there exists a stable partitioning $\mathcal{P}$ of any interval $D$ for $\sigma$ and $\varphi$. \qed
    \end{proposition}

Under a stable partitioning $\mathcal{P}$ for $\sigma$ and $\varphi$, one can discretize $\sigma$ according to truth values indexed by  subformulas $\psi$ of $\varphi$ and intervals $J_i \in \mathcal{P}$.

\begin{definition}[\!\!{\cite[Definition 5]{baeBoundedModelChecking2019}}] 
    Let $\varphi$ be an STL formula, $\mathcal{P}$ be a partitioning of $[0, T]$ and  $\theta\colon \Sub(\varphi) \times \mathcal{P} \to \mathbb{B}$ be an assignment of Boolean values.
    A signal $\sigma: [0, T] \to \mathbb{R}^V$ \emph{matches} the pair  $(\mathcal{P}, \theta)$ if 1)  $\mathcal{P}$ is a stable partitioning for $\sigma$ and $\varphi$, and 2) for each $\psi\in\Sub(\varphi)$ and  $J\in\mathcal{P}$, we have $\theta(\psi, J) = \top$ if and only if  $\sigma^{t} \models \psi$ for each $t\in J$.
\end{definition}
By \cref{prop:stablePartitioningExists}, for any signal $\sigma$ that is finitely variable with respect to $\varphi$, there exists $(\mathcal{P}, \theta)$ that matches it.
We note that it is sufficient to decide the value of $\theta(p, \cdot)$ for atomic propositions $p \in \AP(\varphi)$ in order to identify $\theta$; the values for the other subformulas are then determined by the STL semantics.

\section{Shorthands for Propositional Connectives}\label{sec:shorthand}
We use the following shorthands in our  constraints, where $A,B$ are Boolean variables. They are standard in the MILP community; see e.g.~\cite{wolffOptimizationbasedTrajectoryGeneration2014}.
\begin{equation}\label{eq:MILPBooleanShorthand}
    \begin{array}{lll}
        Z = \lnot A   & \;\text{is short for}\quad\; & Z  = 1 - A                       \\
        Z = A \land B & \;\text{is short for}\quad\; & Z \le A, Z \le B, Z \ge A + B -1 \\
        Z = A \lor B
                      & \;\text{is short for}\quad\; &
        Z \ge A, Z \ge B, Z \le A + B                                                   \\
        A = 0 \;\Rightarrow\; f(x) \ge a
                      & \;\text{is short for}\quad\; &
        f(x) - a \ge M \cdot A                                                          \\
        A = 1 \;\Rightarrow\; f(x) \ge a
                      & \;\text{is short for}\quad\; &
        f(x) - a \ge M \cdot (A - 1)
    \end{array}
\end{equation}
We can also nest the shorthand expressions by introducing auxiliary variables.







Note that one can represent arbitrary nested relation by introducing auxiliary variable.
For example, $Z = A \land (B \lor C)$ is a shorthand notation for $Z = A \land X$ and $X = (B \lor C)$ with a new auxiliary variable $X$.

\section{Rectangular Hybrid Automata (RHAs)}\label{sec:rectHA}
RHAs~\cite{HenzingerKPV95} are restricted hybrid automata and are thus suited to analysis by LP. We briefly review its theory, restricting some definitions for our convenience. We refer to~\cite{HenzingerKPV95,HenzingerK97a} for further details.

Let $V$ be a set of (real-valued) variables. A \emph{rectangular predicate} over $V$ is one of the form $\bigwedge_{x\in V}x\in [a_{x},b_{x}]$ where $a_{x},b_{x}\in \EReal$  are real numbers or $\pm\infty$. We restrict to closed intervals for simplicity.

A \emph{rectangular hybrid automaton (RHA)} $\mathcal{H}$ over a set $V$ of variables is a hybrid automaton (HA) where 1) the flow dynamics at each control mode is described by a rectangular predicate $\bigwedge_{x\in V}\dot{x}\in [a_{x},b_{x}]$ over $\dot{V}\Defeq\{\dot{x}\mid x\in V\}$; 2) the invariant at each control mode is a rectangular predicate over $V$; and 3) each transition  between control modes is labeled with $(\mathsf{grd},\mathsf{update},\mathsf{post})$ where $\mathsf{grd}$ is a rectangular predicate over $V$, $\mathsf{post}$ is a rectangular predicate over $V'=\{x'\mid x\in V\}$ ($x'$ is ``$x$ after transition''); and $\mathsf{update}\subseteq V$ is a subset.

The transition labeled with $(\mathsf{grd},\mathsf{update},\mathsf{post})$ is enabled only when $\mathsf{grd}$ is true, and when it is taken, only the values of $x\in \mathsf{update}$ can be altered. The alteration of the values of $x\in \mathsf{update}$ is fully nondeterministic---the new values can be any reals---although the new values must satisfy $\mathsf{post}$.

The rest of the operational semantics is standard for HAs; this allows us to identify an RHA $\mathcal{H}$ with a system model $\mathcal{M}_{\mathcal{H}}\colon \Signal_{\emptyset}^T\to \pow(\Signal_{V}^T)$ in the sense of \cref{def:sysModelTraceSet}, restricting the domain of signals by some $T$.
RHAs have no input; thus the input variable set is $V'=\emptyset$. The nondeterminism of RHAs is reflected in $\pow$ in the type of $\mathcal{M}_{\mathcal{H}}$.

An example of an RHA is in our \emph{navigation} case study; see \cref{fig:example_nav}.

An encoding of RHAs to MILP is not hard.  For compatibility with the encoding of STL in \cref{sec:encodeSTL}, our encoding of RHAs is in a variable-interval style, too, where the  time domain $\Real_{\ge 0}$ is discretized into intervals $\dotsc, (\gamma_{i-1},\gamma_{i}),\{\gamma_{i}\}, (\gamma_{i},\gamma_{i+1}), \dotsc$. Here all $\gamma_{i}$'s are continuous MILP variables. Due to the restriction  to rectangular predicates---the slope of dynamics is bounded by constants, in particular---the operational semantics of RHAs can be exactly encoded to MILP.

\end{document}